\documentclass[a4paper,onecolumn,accepted=2024-12-14]{quantumarticle}

\pdfoutput=1
\usepackage{graphicx}%
\usepackage{multirow}%
\usepackage{amsmath,amssymb,amsfonts}%
\usepackage{amsthm}%
\usepackage{mathrsfs}%
\usepackage[title]{appendix}%
\usepackage{xcolor}%
\usepackage{textcomp}%
\usepackage{manyfoot}%
\usepackage{booktabs}%
\usepackage{listings}%
\usepackage{ulem}
\usepackage{hyperref}

\usepackage{comment}
\usepackage{tikz}
\usetikzlibrary{positioning,shapes.geometric}

\newtheorem{theorem}{Theorem}[section]
\newtheorem{proposition}[theorem]{Proposition}
\newtheorem{corollary}[theorem]{Corollary}
\newtheorem{lemma}[theorem]{Lemma}

{\bf}{\it}
{\bf}{\it}
{\bf}{\it}
{\bf}{\it}
{\bf}{\it}

{\bf}{\it}
{\bf}{\it}
{\bf}{\it}

\newtheorem{remark}[theorem]{Remark}

\theoremstyle{plain}

\newtheorem{result}{Theorem}

\newtheorem{definition}[theorem]{Definition}%

\newcommand{\one}{\mbox{$1 \hspace{-1.0mm}  {\bf l}$}}
\def\tr{\mathrm{tr}}

\def\none{\tilde{\mbox{$1 \hspace{-1.0mm}  {\bf l}$}}}

\newcommand{\cG}{{\mathcal G}}         

\newcommand{\C}{\mathbb{C}} 
\newcommand{\N}{\mathbb{N}} 

\newcommand{\quadtext}[1]{\quad\text{#1}\quad}

\begin{document}

\title{Asymptotic robustness of entanglement in noisy quantum networks
and graph connectivity}

\author{Fernando Lled\'o}
\affiliation{Departamento de Matem\'aticas, Universidad Carlos III de
Madrid, E-28911, Legan\'es (Madrid), Spain}
\affiliation{Instituto de Ciencias Matem\'aticas, E-28049 Madrid, Spain}
\author{Carlos Palazuelos}
\affiliation{Departamento de An\'alisis Matem\'atico y Matem\'atica Aplicada, Universidad Complutense de Madrid, E-28040 Madrid, Spain}
\affiliation{Instituto de Ciencias Matem\'aticas, E-28049 Madrid, Spain}
\author{Julio I. de Vicente}
\affiliation{Departamento de Matem\'aticas, Universidad Carlos III de
Madrid, E-28911, Legan\'es (Madrid), Spain}
\affiliation{Instituto de Ciencias Matem\'aticas, E-28049 Madrid, Spain}

\begin{abstract}
Quantum networks are promising venues for quantum information processing. This motivates the study of the entanglement properties of the particular multipartite quantum states that underpin these structures. In particular, it has been recently shown that when the links are noisy two drastically different behaviors can occur regarding the global entanglement properties of the network. While in certain configurations the network displays genuine multipartite entanglement (GME) for any system size provided the noise level is below a certain threshold, in others GME is washed out if the system size is big enough for any fixed non-zero level of noise. However, this difference has only been established considering the two extreme cases of maximally and minimally connected networks (i.e.\ complete graphs \textit{versus} trees, respectively). In this article we investigate this question much more in depth and relate this behavior to the growth of several graph theoretic parameters that measure the connectivity of the graph sequence that codifies the structure of the network as the number of parties increases. The strongest conditions are obtained when considering the degree growth. Our main results are that a sufficiently fast degree growth (i.e.\ $\Omega(N)$, where $N$ is the size of the network) is sufficient for asymptotic robustness of GME, while if it is sufficiently slow (i.e.\ $o(\log N)$) then the network becomes asymptotically biseparable. We also present several explicit constructions related to the optimality of these results.
\end{abstract}


\maketitle
\tableofcontents


\section{Introduction}\label{sec1}

Entanglement plays a key role both in the foundations of quantum mechanics and in the applications of quantum information theory. Thus, entanglement theory has been developed over the last decades in order to characterize and quantify this resource and to understand the possibilities and limitations for its manipulation (see e.g.\ the review article \cite{horodecki2009quantum}). However, entanglement is very fragile and distributing an entangled state shared by a large number of parties is a formidable challenge in practice. In the last years, quantum networks have arisen as a very promising platform for large-scale quantum communication \cite{wehner2018quantum}. Additionally, they can be used to model distributed quantum computation schemes \cite{distributedqc1,distributedqc2}. The main idea underlying these structures is that each party (which is identified with a node in the network) controls several qudits, each of which can be entangled with a qudit held by a neighboring node. Thus, each link in the network represents the ability to create bipartite entanglement between the corresponding pair of nodes. Then, the use of quantum repeaters and teleportation makes it possible to establish long-distance quantum communication and to distribute to the parties any multipartite state that might be necessary for some specific quantum information task. This approach needs to consume many instances of the pre-shared entanglement and, thus, the entangled links need to be permanently replenished, which leads to the study of entanglement routing techniques (see e.g.\ \cite{elkouss} and references therein).

However, if entangled links are not refreshed and are modelled by a particular bipartite state, each network pattern defines a multipartite quantum state, which can be understood as a fixed resource state. This state is modelled by an undirected simple graph where the parties correspond to vertices (nodes) and edges (links) represent that a pair of parties share some form of bipartite entanglement to be specified. This view is more natural from the perspective of entanglement theory and leads to the investigation of the entanglement properties of multipartite network states. This approach has been taken for instance in the study of the phenomenon of entanglement percolation \cite{percolation}. Here, one studies which network states enable to establish with non-zero probability a perfect 2-qubit maximally entangled state between a pair of arbitrarily distant parties within the network configuration by local operations and classical communication (LOCC) (see the review \cite{percolationreview} and references therein). Another example is given by the recent work \cite{pairable}, which shows that a $N$-partite network state with $O(\log N)$ qubits per party in a precise configuration gives a construction of a maximally pairable state, i.e.\ for any choice of grouping of the parties into disjoint pairs, a state from which by LOCC one can obtain 2-qubit maximally entangled states shared by each pair. In addition to this, references \cite{yamasaki} and \cite{kraftspee} have studied general LOCC transformations starting from network states where the links correspond to maximally entangled states. On a different note, \cite{networkGMNL} shows that all noiseless network states arising from a connected graph are genuine multipartite nonlocal and \cite{ASGME} uses noisy quantum network states to prove that genuine multipartite nonlocality can be superactivated for any number of parties.

A natural and interesting question in the above direction is to study when a network state displays entanglement. Obviously, the parties corresponding to nodes connected by a link can be entangled, but can this be extended to global entanglement that is spread through the whole network? This notion is precisely captured by the concept of \textit{genuine multipartite entanglement (GME)}. In fact, GME is known to be necessary for maximum sensitivity in quantum metrology \cite{gmesensing1,gmesensing2} and to obtain multipartite secret key \cite{gmekey}. Indeed, most relevant states that appear in applications, such as for example graph states, are GME. Interestingly, since the set of \textit{biseparable (BS)} states (i.e.\ non-GME states) is closed by LOCC, any network state that is not GME is severely limited for many applications as it is impossible to obtain any GME state from it by further LOCC processing.

It is easy to see that any network state corresponding to a connected graph and where the edges correspond to pure bipartite entangled states has to be GME. However, an interesting phenomenon arises when the nodes share mixed states. This can already be seen by considering the simple and standard noise model in which every edge describes an isotropic state shared by the adjacent vertices with visibility $p$ (where $p=1$ corresponds to the ideal case of a pure maximally entangled state; see \cite{isotropic}).
In \cite{ASGME} it has been shown that the properties of the graph codifying the structure of the network play a non-trivial role in this problem. To wit, it turns out that for any constant $p<1$ there exist graph sequences such that the corresponding network state is not GME for a large enough size. That is, the slightest form of noise renders the network state biseparable if the number of parties is sufficiently large.  Remarkably, on the other hand, there also exist graph sequences for which the corresponding network state has an asymptotically robust GME, i.e.\ GME can be displayed for any system size for some fixed $p<1$. Thus, if the devices that produce entangled pairs between nodes are sufficiently developed as to guarantee a visibility above a certain threshold, one can find network configurations that are GME for any number of parties. On the other hand, other network configurations are doomed to failure and fundamentally limited in applications since in any realistic scenario the entangling devices will always operate outside the ideal case of a perfect visibility ($p=1$).

Intuitively, the more connected a network is the more likely it is to have robust multipartite entanglement properties. In fact, \cite{ASGME} proved the existence of these two different behaviors by only considering the extreme cases of tree graphs (minimally connected) and completely connected graphs (maximally connected). This, however, leaves wide open a finer analysis of the relation between asymptotic robustness of GME in noisy network states of increasing size and different notions of graph connectivity. In this article, we combine quantum information and graph theory techniques to study systematically and in more depth this problem, unveiling a subtle and non-trivial relation. The most standard measures of the connectivity of a graph are the vertex-connectivity and the edge-connectivity \cite{diestel,bondy}. While we find that when these quantities grow sufficiently fast this ensures that the network has asymptotically robust GME, on the other hand, it turns out that minimally connected graph sequences as measured by these two parameters can also have this property. Our analysis reveals that the degree growth of the graph sequence, which is also intimately connected to the connectivity of the graphs, gives the strongest results providing sufficient conditions for both asymptotic GME and asymptotic biseparability. Yet, we prove that this parameter alone cannot fully characterize these properties.

Although precise definitions will be given below (see Section~\ref{sec2}), we will briefly introduce here the minimum necessary notions for understanding the main results of the article. We consider simple, connected and undirected graphs $G=(V,E)$ with order $|G|=|V|$. Given a vertex $v\in V$ we denote its degree (the number of vertices adjacent to $v$) by $\delta(v)$ and define the maximal and minimal degrees of $G$ respectively by
\begin{equation}
\delta_{\mathrm{min}}(G)=\min_{v\in V}\delta(v)
\quadtext{and}
\delta_{\mathrm{max}}(G)=\max_{v\in V}\delta(v).
\end{equation}
Given a graph $G=(V,E)$ as above, a dimension $d\in\mathbb{N}$ ($d\geq2$) and a visibility parameter $p\in[0,1]$,  the \textit{isotropic network state} is the $|G|$-partite density matrix given by
\begin{equation}
\sigma_d(G,p)= \bigotimes_{e\in E}\rho_e(p,d),
\end{equation}where the subscript $e=\{u,v\}$ denotes that the isotropic state $\rho(p,d)$ (cf.\ Eq.\ (\ref{isotropicd}) below) is shared by the parties $u$ and $v$. The state $\sigma_d$ corresponds to a
$|G|$-partite state described by the graph $G$ where the vertices are identified with the parties and the edges are identified with bipartite isotropic states of parameter $p$, shared by the corresponding adjacent parties. Thus, for any given values of $p$ and $d$ we can describe any sequence of isotropic network states with growing size by a sequence $\cG=(G_n)_{n\in\N}$ of (finite, simple and connected) graphs $G_n=(V_n,E_n)$ such that the order
$|G_n|$ tends to infinity. The multipartite entanglement robustness properties that we want to study are precisely captured by the notions of asymptotically genuine multipartite entanglement (AGME) and asymptotic biseparability (ABS), as introduced in Definition \ref{def:asymptotic} below and that rigorously establish the intuitive ideas discussed in the previous paragraphs.

The main results in this article (where we use standard asymptotic notation) give sufficient conditions guaranteeing either AGME or ABS behaviors, which go well beyond the mere existence of the phenomenon that was established in \cite{ASGME}.

\

\begin{result}[High degree growth implies AGME;  cf. Theorem~\ref{th:degAGME}]
Any graph sequence $\cG=(G_n)_{n\in\N}$ such that $\delta_{\mathrm{min}}(G_n)\in\Omega(|G_n|)$ is AGME.\label{thm:1}
\end{result}
\

\begin{result}[Low degree growth implies ABS; cf. Theorem~\ref{th:degABS}]
Any graph sequence $\cG=(G_n)_{n\in\N}$ such that $\delta_{\mathrm{max}}(G_n)\in o\left(\log|G_n|\right)$ is ABS.\label{thm:2}
\end{result}

\

We complete the analysis with explicit constructions of graph sequences that show the sharpness of these results and the limitations of the degree growth to fully describe the AGME and ABS properties.

\

\begin{result}[Sublinear degree growth is compatible with ABS; cf. Theorem~\ref{th:highdegABS}]

For any function $f\colon\mathbb{N}\to\mathbb{N}$ such that $f(n)\in o(n)$, there exists an ABS graph sequence $\mathcal{G}=(G_n)_{n\in\mathbb{N}}$ such that
$\delta_{\mathrm{min}}(G_n)\in\Omega(f(|G_n|))$.\label{thm:3}
\end{result}

\

\begin{result}[Sublinear degree growth is compatible with AGME; cf. Theorem~\ref{th:lowdegAGME}]
For any $\alpha\in(0,1]$, there exists an AGME sequence of regular graphs $\cG=(G_n)_{n\in\N}$ such that $\delta_{\mathrm{max}}(G_n)\in O(|G_n|^\alpha)$.\label{thm:4}
\end{result}

\

Notice that the degree is directly related to the local dimension of the network state. Thus, Theorem~\ref{thm:2} proves that AGME requires that this quantity grows as the network gets bigger (and, in fact, it has to grow sufficiently fast). On the other hand, Theorem~\ref{thm:1} shows that if it grows fast enough, then AGME is guaranteed and Theorem~\ref{thm:3} shows that this latter condition cannot be improved. Theorem~\ref{thm:4} gives instances of asymptotically robust GME networks with the lowest overhead in terms of local dimension that we have been able to find and together with Theorem~\ref{thm:2} leaves only a relatively small gap for improvements. Moreover, Theorems~\ref{thm:3} and
\ref{thm:4} prove that the AGME and ABS behaviors cannot be fully characterized only in terms of the degree. In addition to this, we relate these notions to the connectivity of the graph sequence as measured by the vertex and edge connectivities, obtaining this same conclusion. In fact, when these quantities grow sufficiently fast they also guarantee AGME but give a weaker condition to that of Theorem~\ref{thm:1} (however, this serves as an important technical ingredient in the proofs of our main results). On the contrary, they can never give a sufficient condition for ABS. This is because we can construct minimally connected graph sequences which are nevertheless AGME (see Proposition~\ref{th:lowlambda} and Fig.~\ref{fig:mancuerna}). Thus, our analysis indicates that the AGME and ABS properties represent an alternative notion of connectivity in graphs that can be related to these and other relevant graph parameters, but that seems to be ultimately independent of them. We therefore think that this problem provides a fruitful mathematical interplay between entanglement theory and graph theory that deserves further study.

The article is organized as follows: in Section~\ref{sec2} we present basic material on multipartite quantum systems and graph theory needed later. In the following section we introduce in Definition~\ref{def:asymptotic} the notions of AGME and ABS and provide different characterizations and results in relation to these concepts.
In Section~\ref{sec3} we present a sufficient condition for AGME that reduces the problem to graph theoretic notions. In the next section we study the relation of AGME with vertex and edge connectivity and establish some results that we use later in the proof of our main results. In Section~\ref{secdeg} we obtain the sufficient conditions for AGME and ABS in terms of the degree growth introduced before as Theorems~\ref{thm:1} and \ref{thm:2} and in Section~\ref{subsecdeg3} we provide the explicit constructions that lead to Theorems~\ref{thm:3} and \ref{thm:4} above. We conclude with a discussion on our results and possible future research directions.

\section{Basic definitions and properties of quantum states and networks}\label{sec2}

In this section we review the notions of quantum theory and graph theory that will be used throughout the article. In our analysis we use as well standard asymptotic notation. Given two functions $f\colon\mathbb N\rightarrow\mathbb{R}$ and $g\colon\mathbb N\rightarrow[0,\infty)$ (such that $g(n)>0$ for large enough $n$) we write $f(n)\in O(g(n))$ if there exists a real constant $c>0$ and an integer $n_0\geq 1$ such that $|f(n)|\leq cg(n)$ for every $n\geq n_0$. We say that $f(n)\in \Omega(g(n))$ if $g(n)\in O(f(n))$. Equivalently, $f(n)\in O(g(n))$ (resp.\ $f(n)\in \Omega(g(n))$) if and only if (iff) $\limsup_{n\to\infty}(|f(n)|/g(n))<\infty$ (resp.\ $\liminf_{n\to\infty}(f(n)/g(n))>0$).
We write $f(n)\in o(g(n))$ if for every real constant $c>0$ there exists an integer $n_0\geq 1$ such that $|f(n)|\leq cg(n)$ for every $n\geq n_0$ and $f(n)\in \omega(g(n))$ if $g(n)\in o(f(n))$. Equivalently, $f(n)\in o(g(n))$ (resp.\ $f(n)\in \omega(g(n))$) iff $\lim_{n\to\infty}(f(n)/g(n))=0$ (resp.\ $\lim_{n\to\infty}(f(n)/g(n))=\infty$).

\subsection{Multipartite quantum systems}\label{sec21}

In this subsection we introduce the basic definitions and results on quantum states that will be needed in the article. Some standard references are \cite{Nielsen,Watrous}.

We will consider here finite-dimensional quantum systems. The (complex) Hilbert spaces $H$ will have dimension $d\in\{2,3,\dots \}$ and we denote by
$\{|i\rangle\}_{i=1}^d$ an orthonormal basis of $H$. Given $H$, quantum states are characterized by density matrices, that is, by linear operators $\rho\colon H\to H$ such that $\rho$ is positive semidefinite (from now on, $\rho\geq 0$) and $\tr\rho=1$.
The set of all density matrices will be denoted by $D(H)\subset B(H)$, where $B(H)$ is the set of all bounded linear operators on $H$. Quantum channels $\Lambda\colon D(H)\to D(H')$ are given by completely positive and trace preserving linear maps. Given $\rho,\sigma\in D(H)$, their \textit{fidelity} is defined by
\begin{equation}\label{fidelity}
F(\rho,\sigma)=\tr^2\sqrt{\sqrt{\rho}\sigma\sqrt{\rho}}.
\end{equation}

For bipartite quantum systems, the associated Hilbert space is $H=H_1\otimes H_2$. A state with density matrix $\rho\in D(H)$ is \textit{separable} if it can be written as
\begin{equation}\label{def:separable}
\rho=\sum_{j}p_j\sigma_j\otimes\tau_j,
\end{equation}
where for all $j$, $p_j\geq0$ with $\sum_jp_j=1$, and $\sigma_j\in D(H_1)$, $\tau_j\in D(H_2)$. The set of \textit{separable density matrices} in $D(H)$ is denoted by $S(H)$. Density matrices that are not separable are known as \textit{entangled}. It is well known that the set of separable states is closed under LOCC quantum channels (see \cite{locc} for a precise definition), i.e.\ $S(H)$ is LOCC-stable. We will often consider the following particular state $\rho(p,d)\in D(H_1\otimes H_2)$ with $\dim H_1=\dim H_2=d\in\mathbb{N}$ ($d\geq2$), which is known as the \textit{isotropic state} \cite{isotropic} and is given by
\begin{equation}\label{isotropicd}
\rho(p,d)=p\phi^+_d+(1-p)\none_{d^2},
\end{equation}
where $\one_{d^2}$ is the identity operator on $H=H_1\otimes H_2$ and
$\none_{d^2}=\one_{d^2}/d^2$ denotes the normalized identity operator.
 Moreover,
\begin{equation}\label{maxent}
|\phi^+_d\rangle=\frac{1}{\sqrt{d}}
 \sum_{i=1}^d|ii\rangle\in H_1\otimes H_2
 \quadtext{and}
 \phi^+_d=|\phi^+_d\rangle\langle\phi^+_d|\in D(H_1\otimes H_2)
\end{equation}
is the rank~$1$ projection corresponding to the $d$-dimensional \textit{maximally entangled state}. Here and in the following, we will use the short-hand $|ij\rangle=|i\rangle\otimes|j\rangle$. The parameter $p\in[0,1]$ in Eq.~(\ref{isotropicd}) is often referred to as
\textit{visibility} and it might be understood as a measure of the quality with which a maximally entangled state is prepared when our set-up introduces white noise. Notice that the \textit{completely depolarizing quantum channel} in $D(H_1\otimes H_2)$ given by
$\Lambda(X)=\tr(X)\none_{d^2}$, is an LOCC channel and that the set of LOCC channels is convex \cite{locc}. Thus, whenever $p'\leq p$ there always exists an LOCC channel $\Lambda$ such that $\Lambda(\rho(p,d))=\rho(p',d)$. Since $S(H)$ is a LOCC-stable subset of the state space, $\rho(1,d)$ is entangled and $\rho(0,d)$ is separable for every dimension $d$, we conclude that there is a \textit{critical visibility} $\overline{p}$ separating, for each isotropic state, entanglement from separability. Actually, it is shown in \cite{isotropic} that $\rho(p,d)$ is entangled iff
\begin{equation}\label{isotropicthreshold}
p>\frac{1}{d+1}.
\end{equation}

Any state $\rho\in D(H_1\otimes H_2)$ has a convex decomposition as
\begin{equation}\label{bsa1}
\rho=\lambda\sigma+(1-\lambda)\tau\;,
\end{equation}
where $\sigma\in S(H_1\otimes H_2)$ is separable, $\tau\in D(H_1\otimes H_2)$ and $\lambda\in[0,1]$. The notion of \textit{best separable approximation} (see \cite{BSA}) is defined as
\begin{equation}\label{bsa2}
    \mathrm{BSA}(\rho)=\max\left\{\lambda\in[0,1]\mid \rho=\lambda\sigma+(1-\lambda)\tau, \sigma\in S(H), \tau\in D(H)\right\}.
\end{equation}
The quantity $E_{\mathrm{BSA}}(\rho)=1-\mathrm{BSA}(\rho)$ is an \textit{entanglement measure} and it cannot increase by LOCC transformations (see \cite{pleniovirmani}).
In particular, results in \cite{BSAisotropic} show that for isotropic states given in Eq.~(\ref{isotropicd})
\begin{equation}\label{bsaisotropic}
\mathrm{BSA}(\rho(p,d))=\frac{d+1}{d}(1-p)\;,
\quadtext{for} p\geq 1/(d+1)\;.
\end{equation}

For \textit{multipartite quantum systems} with $N$ parties, we denote
$[N]=\{1,2,\ldots,N\}$ and the corresponding Hilbert space is $H=\bigotimes_{i=1}^N H_i$.
\begin{definition}
A density matrix $\chi\in D(H)$ on $H=\bigotimes_{i=1}^N H_i$ is said to be separable in the nontrivial
bipartition $M|\overline{M}$, (where $\emptyset\not=M\subsetneq[N]$ and
$\overline{M}$ is the complement of $M$ in $[N]$),
if it has a convex decomposition of the form
\[
 \chi=\sum_{j}p_j\sigma_j\otimes\tau_j\,,
 \quadtext{with} p_j\geq0 \quadtext{and}\sum_jp_j=1\;,
\]
and where $\sigma_j\in D\left(\bigotimes_{i\in M}H_i\right)$ and
$\tau_j\in D\left(\bigotimes_{i\notin M}H_i\right)$ for all $j$.
The set of all density matrices with this property is denoted by $S_M(H)$. The set of \textbf{biseparable density matrices (BS)}
is the convex hull of the union of these sets
\begin{equation}\label{def:bs}
\mathrm{BS}(H)=\mathrm{conv}\left\{\bigcup_{\emptyset\neq M\subsetneq[N]}S_M(H)\right\}\,.
\end{equation}
A state that is not biseparable is called \textbf{genuine multipartite entangled (GME)}.
\end{definition}
Note that a density matrix $\rho\in D(H)$ is biseparable iff
\begin{equation}\label{def:biseparable}
\rho=\sum_{\emptyset\neq M\subsetneq[N]}p_M\chi_M,
\end{equation}
where for any $M$ we have $p_M\geq0$ with $\sum_Mp_M=1$ and $\chi_M\in S_M(H)$. Notice that it suffices to consider in Eqs.~(\ref{def:bs}) and (\ref{def:biseparable}) only subsets $M$ that give rise to different bipartitions, but to simplify the notation we do not impose this constraint explicitly. It follows from the LOCC-stability of the set of separable states that the set of biseparable states is LOCC-stable as well.
In the proof of the main result in Section~\ref{subsecdeg1} it will be convenient to introduce the particular notion of \textit{$1$-biseparable states}, i.e., biseparable states where the convex decomposition in
Eq.~(\ref{def:biseparable}) runs over all subsets $M$ with cardinality one, i.e., $|M|=1$. The set of states which are separable in the bipartition
specified by $M=\{j\}$ are denoted simply by $S_j(H)$ and, hence, the set of 1-biseparable states is $\mathrm{conv}\{\bigcup_{j\in[N]}S_j(H)\}$.

Given then an $N$-party Hilbert space $H=\bigotimes_{i=1}^N H_i$ and
a non-empty subset $M\subsetneq[N]$, the \textit{partial trace} with respect to $\bigotimes_{i\in M}H_i$ (or simply with respect to $M$) is defined to be the linear operator
\begin{equation}
\tr_{(\bigotimes_{i\in M}H_i)}\colon B(H)\to
B\big(\otimes_{i\notin M}H_i\big)
\end{equation}
such that
\begin{equation}
\tr_{\left(\bigotimes_{i\in M}H_i\right)}\left(\bigotimes_{i=1}^NA_i\right)=\tr\left(\bigotimes_{i\in M}A_i\right)\bigotimes_{i\notin M}A_i,
\end{equation}
for any $A_i\in B(H_i)$.

\subsection{Quantum networks and graph theory}\label{sec22}

Since graph theoretical concepts play an important role in our analysis we begin recalling some basic notions and results on graphs needed later
(see e.g., \cite{diestel,bondy}).

\subsubsection{Graph theoretic aspects}\label{subsec:GT}
In this article, a graph is a finite and simple graph (no loops nor multiple edges are allowed) given by a pair $G = (V, E)$ consisting of two disjoint finite sets
$V$, the set of vertices or nodes, and $E$ the set of edges or links.
The elements of $E$ are unordered pairs $\{u, v\}$ of vertices $u, v \in V$ which we will also denote by
$e,f\dots$.
If $e=\{u,v\} \in E$, we say that $u$ is adjacent to $v$ (or that $u$ is a
neighbour of $v$ and we denote the vertex neighborhood of $v$ by $N(v)$); moreover, we say that $e$ is adjacent to $u$ (and to $v$).
The cardinality of the vertex set is the order of the graph which we denote by $|G|=|V|$. Sometimes it will be convenient to numerate the vertices and identify $V=[|G|]$.
A path is a simple graph whose vertices can be arranged in a linear sequence so that two vertices are adjacent if they are consecutive in the sequence and are nonadjacent otherwise. A graph $G$ is connected if any pair of vertices can be linked by a path in $G$. A maximal connected subgraph of $G$
is called a connected component of $G$.
The distance $d_G(u,v)$ between two vertices $u,v\in V$ of a connected graph is the minimal length of the paths in $G$ joining $u$ with $v$. The diameter of the graph is then defined by
\begin{equation}
 \mathrm{diam}(G):=\max_{u,v\in V}d_G (u,v)\,.
\end{equation}

Given two subsets of vertices $V_1,V_2\subset V$, we denote the set of $V_1$-$V_2$ edges by
\begin{equation}
E(V_1,V_2):=\{ \{v_1,v_2\} \mid v_i \in V_i \;,\; i=1,2 \}
 \quadtext{and}
 E(V_1):=E(V_1,V_1);.
\end{equation}
We also denote, for simplicity, all edges adjacent to a vertex $v\in V$ as
\begin{equation}
 E(v):=E\big(\{v\},V\big)\quadtext{and define the degree of $v$ by} \delta(v):=|E(v)|\;.
\end{equation}
The minimal and maximal degree of $G$ are given respectively by
\begin{equation}\label{maxmindegree}
\delta_{\mathrm{min}}(G)=\min_{v\in V}\delta(v)
\quadtext{and}
\delta_{\mathrm{max}}(G)=\max_{v\in V}\delta(v).
\end{equation}
A subgraph $G'$ of a graph $G$ is a graph formed from a subset of the vertices and of edges of $G$. A subgraph is called \textit{spanning} if
$V'=V$, i.e., it contains all vertices from $G$.
Edge deletion is a natural procedure to obtain a subgraph from a given graph. An induced subgraph $G'$ of $G$ is a subgraph that includes all the edges with adjacent vertices in $V'$. Given the vertex subset
$V'\subset V$ the subgraph induced by $V'$ is the subgraph $G'=(V',E')$ with
$E'=E(V')$. Vertex deletion of $G$ specifies an induced subgraph.

In this article, the graph connectivity of the graph associated to a network state will play an important role. There are two basic notions of connectivity for a graph $G$:
the vertex-connectivity $\kappa(G)$ and the edge-connectivity $\lambda(G)$
(see e.g. \cite{diestel,bondy}).
Recall that an edge-cut of $G$ (and similarly for a vertex-cut) is a subset of edges that, when removed, the remaining graph has at least two connected components. If a single edge provides and edge-cut of the graph we call it
bridge-edge. The edge-connectivity of a
graph is the size of a smallest edge-cut. Formally, for any graph $G$ such that $|G|>1$ we define
\begin{align}
\kappa(G) &=\min\left\{|W| \Bigm\vert  G'=\left(V\setminus W,E \setminus
\left(\bigcup_{v\in W}E(v)\right)\right)\textrm{ is disconnected or } |G'|=1
\right\}
\\
\lambda(G) &=\min\left\{|F| \mid G'=(V,E\backslash F)\textrm{ is disconnected}\right \}.
\end{align}
It is a well-known fact that $\delta_{\mathrm{min}}(G)\geq\lambda(G)\geq\kappa(G)$. Since large vertex-connectivity implies large
edge-connectivity, we will focus in this article on the latter notion. The complete graph of order $n$, $K_n=(V_n,E_n)$, can de defined by setting $|V_n|=n$ and $E_n=\{V'_n\subset V_n:|V'_n|=2\}$ (i.e.\ each pair of vertices is connected by an edge) and a tree is a graph in which every pair of distinct vertices is connected by exactly one path. Intuitively, the complete graph should be the most connected graph and trees should be poorly connected. Indeed, the maximal value of the edge-connectivity for simple graphs of order $n$ is attained only by the complete graph with $\lambda(K_n)=n-1$, while trees provide examples of connected graphs with minimal edge-connectivity, i.e.\ $\lambda(G)=1$. Graphs such that $\delta_{\mathrm{min}}(G)=\lambda(G)$ are usually referred to as
maximally edge-connected. It is known
(see \cite{chartrand}) that if
\begin{equation}\label{maxedgeconnected}
\delta_{\mathrm{min}}(G)\geq\left\lfloor\frac{|G|}{2}\right\rfloor
\end{equation}
holds, then $G$ is maximally edge-connected. Connectivity has also a local formulation.
Given two distinct vertices $u,v\in V$, denote by $\lambda_G(u,v)$ the maximal number of edge-disjoint paths (i.e., paths not sharing any edge) that connect $u$ and $v$. Menger's theorem (see e.g. \cite[Theorem~3.3.1]{diestel}) provides a dual characterization for the connectivities. In the case of the edge-connectivity, it states that
 \begin{equation}\label{menger}
 \lambda(G)=\min_{u\not= v\in V}\lambda_G(u,v).
 \end{equation}

\subsubsection{Noisy quantum networks}\label{subsed:noisy}
As mentioned before, quantum network states are characterized by a graph. Each vertex represents a party, while each edge represents that the corresponding parties share a bipartite state. In our model, we consider for simplicity that all edges stand for an isotropic state of the same given dimension $d$ and visibility $p$. Nevertheless, our techniques could be used in principle for a much more general class of models. Thus, given a graph $G=(V,E)$, a dimension $d\in\mathbb{N}$ ($d\geq2$) and a visibility parameter $p\in[0,1]$, the \textit{isotropic network state} is a $|G|$-partite density matrix on $H=\bigotimes_{v\in V}H_v=\bigotimes_{i=1}^{|G|}H_i$ (recall that, as we have mentioned in the beginning of Sec.~\ref{subsec:GT}, we will often identify $V$ with $[|G|]$) given by
\begin{equation}\label{networkGdp}
\sigma_d(G,p)= \bigotimes_{e\in E}\rho_e(p,d),
\end{equation}
where the subscript $e=\{u,v\}$ denotes that the isotropic state is shared by the
parties $u$ and $v$ (see Fig.~\ref{fig:split}). The analysis of the local structure of the Hilbert spaces can be refined by noting that for each edge $e=\{u,v\}$ the Hilbert space of the isotropic state
$\rho_e(p,d)\in D(H_e)$ can be split as
$H_e=H_{ve}\otimes H_{ue}$ and, for each
$v\in V$, we have that
\begin{equation}\label{eq:ie}
H_v=\bigotimes_{e\in E(v)}H_{ve}
\quadtext{and}
\dim H_v=d^{\delta(v)}\;.
\end{equation}
By the handshaking lemma we then have that
$
 \dim(H)=d^{\left(\sum_{i=1}^N \delta(v_i)\right)}=d^{2|E|}\;.
$
\vspace{-0.2cm}

\begin{figure}[h]
 \centering
 \begin{tikzpicture}[scale=0.70]
    \node[draw, circle, inner sep=1pt, minimum size=4pt, fill=black, label=below:$u$] (u) at (-2,-2) {};

    \node[draw, circle, inner sep=1pt, minimum size=4pt, fill=black, label=below:$w$] (w) at (2,-2) {};

    \node[draw, circle, inner sep=2pt, minimum size=13pt, fill=white, label=below:$v$] (v) at (0,0) {};

    \node[draw, dashed, dash pattern=on 2pt off 2pt, circle, inner sep=1pt, minimum size=5pt, fill=white] (v1) at (-0.15, 0) {};
    \node[draw, dashed, dash pattern=on 1.5pt off 1.5pt, circle, inner sep=1pt, minimum size=5pt, fill=white] (v2) at (0.15, 0) {};

    \draw (u) -- node[midway, above left] {$e_1$} (v1);
    \draw (w) -- node[midway, above right] {$e_2$} (v2);
\end{tikzpicture}
\caption{Isotropic network state of three parties specified by a path graph. The vertices $u,v,w$ specify the parties and the edges $e_1$ and $e_2$ that the corresponding pair of parties share a bipartite isotropic state. Notice that $u,v$ represented in black hold a single qudit, while party $v$ holds two qudits coming from different pairs of isotropic states. We represent this by the two white-dashed circles inside a white vertex representing $v$.}
\label{fig:split}
\end{figure}

We will often use later that that if $G'$ is the subgraph of
$G$ corresponding to deleting all edges in $F\subset E$ and all vertices $v$
such that $E(v)\subset F$, it then follows that
\begin{equation}\label{edgedeletion}
    \tr_{\left(\bigotimes_{e\in F}H_e\right)}\left(\sigma_d(G,p)\right)=\sigma_d(G',p)\;.
\end{equation}

Note that if the graph $G=G^{(1)}\sqcup G^{(2)}$ has two (disjoint) connected components then, for every $p$ and $d$, $\sigma_d(G,p)$ is trivally BS as it is separable in a bipartition specified by the vertices of one of the connected components. If we study this question for a connected graph $G$ as a function of the noise parameter $p$ and for a given $d$ we always have some robustness for both GME and BS. If for some party $i$ it holds that $\rho_e\in S(H_e)$ for all $e\in E(i)$, then $\bigotimes_{e\in E}\rho_e\in S_i(H)$ and it is therefore biseparable. Thus, it follows from Eq.~(\ref{isotropicthreshold}) that $\sigma_d(G,p)$ is biseparable if $p\leq1/(d+1)$. On the other hand, it can be easily shown that, under the assumption that the graph is connected, $\sigma_d(G,p)$ is always GME if $p=1$ (see \cite{ASGME}). Since the set of biseparable states is closed (in the norm topology),
we then have that for every given $d$ and $G$, there exists $\epsilon>0$ such that the GME property stays true for $p>1-\epsilon$.

We conclude this section recalling some important properties of isotropic network states that will be needed later.

\begin{proposition}\label{prop:facts}
Let $G$ be a graph and consider a fixed dimension $d$. Then
\begin{itemize}
 \item[(i)] \label{fact1}
 there exists $\epsilon>0$ such that $\sigma_d(G,p)$ is GME if $p>1-\epsilon$;
 \item[(ii)] \label{fact2}
 if $\sigma_d(G,p)$ is BS, then $\sigma_d(G,p')$ is BS for all $p'\leq p$.
\end{itemize}
\end{proposition}

Part (i) follows from the fact that $\sigma_d(G,1)$ is GME for all $d$ for every connected graph $G$ and that the set of biseparable states is closed.
Part (ii) shows a hereditary property of biseparability for isotropic states
in relation with the visibility parameter. It
holds because $\rho(p,d)$ can be transformed by LOCC into $\rho(p',d)$ for all
$p'\leq p$ and because the set of BS states is LOCC-stable. Taking into account the above results, given any graph $G$ and a fixed dimension $d$,
we introduce the critical visibility for isotropic network states (see Fig.~\ref{fig:criticalvis})
\begin{equation}\label{eq:pbar}
\overline{p_d(G)}:=\max\{p\mid \sigma_d(G,p)\in \mathrm{BS}(H)\}
\quadtext{and one has}\overline{p_d(G)}\in[1/(d+1),1)\;.
\end{equation}
The lower bound of the interval is a consequence of Eq.~(\ref{isotropicthreshold}) and the upper bound follows from the fact that $\sigma_d(G,1)$ is GME. In particular, given a graph sequence $\cG=(G_n)_n$ and a dimension $d$, the corresponding sequence
$\left(\overline{p_d(G_n)}\right)_n$ is bounded.

\begin{figure}[h]
\centering
\begin{tikzpicture}[scale=0.7]
    \draw[thick] (0,0) -- (7.5,0);
    \draw[ultra thick] (7.5,0) -- (10,0);

    \draw[thick] (0,0.2) -- (0,-0.2);   
    \draw[thick] (10,0.2) -- (10,-0.2); 
    \draw[thick] (3.5,0.2) -- (3.5,-0.2); 
    \draw[thick] (7.5,0.2) -- (7.5,-0.2); 

    \node[below] at (0,-0.2) {0};
    \node[below] at (10,-0.2) {1};

    \node[above] at (1.75,0.2) {\small BS};
    \node[above] at (5.25,0.2) {\small BS};
    \node[above] at (8.75,0.2) {\small {\bf GME}};

    \node[below] at (3.5,-0.2) {$\frac{1}{d+1}$};
    \node[below] at (7.5,-0.2) {$\overline{p(G)}$};
\end{tikzpicture}
 \caption{Critical visibility dividing the isotropic network state given by the graph $G$ into BS and GME regions as a function of $p\in[0,1]$. More specifically, the state is biseparable when $p\leq\overline{p_d(G)}$ (and GME otherwise) and we know, in addition, that $\overline{p_d(G)}\in[1/(d+1),1)$.}
\label{fig:criticalvis}
\end{figure}


\section{Asymptotic BS and GME}\label{sec:asymptotic}

In the preceding subsection we already considered certain robustness properties of the network state $\sigma_d(G,p)$ for a given graph $G$ and dimension $d$
as a function of the visibility $p$. In this article, however, we are interested in a different notion of robustness where the parameters $d$ and $p$ are fixed but the graph $G$ is allowed to change.
This notion captures the more practical situation in which the current state of technology provides us with devices that prepare isotropic states in a certain dimension with a certain visibility, which will always be smaller than one. The question we are interested in is which network patterns we should aim at in order to have a GME state under a promise in the visibility. If the number of parties is fixed, given the preceding observations and the fact that the number of simple connected graphs of a given order is always finite, we know that if the experimentalists work hard enough to raise the visibility above a certain threshold we can always produce a GME state of this number of parties in this way. However, the technological applications of quantum networks place no bound on their size and we might want to build them bigger with time. Thus, given $d$, is there a value  $p_{0}$ of the visibility, such that there always exists an isotropic network state of any number of parties which is GME if $p\geq p_{0}$? If so, which network configurations have this property? Notice then that the study of asymptotic robustness of GME boils down to the specification of a graph sequence $\cG=(G_n)_{n\in\N}$ that describes the network as the size grows. For simplicity (and to avoid repetitions in definitions, theorems etc.) we will simply call a \textit{graph sequence} $\cG$ a sequence of \textit{finite, simple, connected} graphs satisfying that $\lim_n|G_n|=\infty$.

\begin{definition}\label{def:asymptotic}
 Consider a graph sequence $\cG=(G_n)_{n\in\N}$.
\begin{itemize}
 \item[(i)]
 $\cG$ is \textbf{asymptotically genuine multipartite entangled (AGME)} in a given dimension $d$ if there exists $p_0\in(1/(d+1),1)$ such that for all $p\geq p_0$ and for all $n\in\N$ the isotropic network state $\sigma_d(G_n,p)$ is GME. If the graph sequence $\cG$ is AGME for any given $d\geq2$, we say simply that $\cG$ is AGME.

 \item[(ii)]
$\cG$ is \textbf{asymptotically biseparable (ABS)} in a given dimension $d$ if for all $p\in(1/(d+1),1)$ there exists $n_0\in\mathbb{N}$ such that for all $n\geq n_0$ the state $\sigma_d(G_n,p)$ is BS. If $\cG$ is ABS for any given $d\geq2$, we say that $\cG$ is ABS.
\end{itemize}
\end{definition}

It follows immediately from these definitions that the asymptotic notions of BS and GME are stable under taking subsequences and we will often make use of this fact. Notice that, given a graph sequence $\cG=(G_n)_{n\in\N}$, any subsequence can be specified by a strictly increasing function $s\colon\N\to\N$ and it will be denoted by $\cG_s=(G_{s(n)})_{n\in\N}$.

ABS graph sequences come with a clear practical limitation. Independently of how evolved technology is, in practice it will always hold that $p<1$. Thus, no matter how hard experimentalists work to raise this threshold, there will always be an upper bound to the number of parties that can be prepared in a GME state in this way. On the other hand, if a graph sequence is AGME, even though $p_0$ might be very close to 1, GME is asymptotically robust: if the devices are sufficiently improved so as to fulfill that $p \geq p_0$, then following the connection pattern given by $\cG$ we can build a network that is GME for any system size. Remarkably, it is not only unclear a priori how to verify these properties for a given graph sequence, but there is no immediate reason in principle to ascertain what the general picture here is. Do some graph sequences have the AGME property while others do not or does one of the two extreme cases hold in which all graph sequences are AGME or all graph sequences are ABS? The answer to this question is the main result of \cite{ASGME}: it turns out that a sequence of complete graphs of increasing order is AGME, while any sequence of tree graphs of increasing order is ABS. From the theorist point of view, this is the most interesting scenario: AGME is possible, but not trivially possible. Thus, despite answering the above question, \cite{ASGME} leaves completely open which graph sequences are robust outside of the two extreme cases mentioned above. In this work we present a much more in-depth analysis of which graph sequences lead to the AGME or the ABS properties. For this purpose we present next some equivalent characterizations that stem from the above definitions.

\begin{proposition}\label{characterizationAGME}
Let $\cG=(G_n)_n$ be a graph sequence and denote by
$\overline{p_d(G_n)}\in\left[1/(d+1),1\right)$ the
corresponding sequence of threshold visibilities. Then, the following conditions are
equivalent:
\begin{itemize}
\item[(a)] $\cG$ is AGME in a given dimension $d$;
\item[(b)] there exists $p_0\in(1/(d+1),1)$ such that for any $n\in\N$ the state
$\sigma_d(G_n,p_0)$ is GME;
\item[(c)] there exist $ p_0\in(1/(d+1),1)$ and $ n_0\in\mathbb{N}$ such that
for all $n\geq n_0$ the state $\sigma_d(G_n,p_0)$ is GME;
\item[(d)] the sequence of threshold visibilities satisfies
$\sup_{n}\overline{p_d(G_n)}<1\;;$
\item[(e)] the sequence threshold visibilities satisfies
$\limsup_{n\to\infty}\overline{p_d(G_n)}<1\;.$
\end{itemize}
\end{proposition}

\begin{proof} The chain of implications (a) $\Rightarrow$ (b) $\Rightarrow$ (c) is immediate. To show that (c) $\Rightarrow$ (d), recall that the assumption implies that there exists an $ n_0\in\mathbb{N}$ such that $\overline{p_d(G_n)}<p_0<1$ for all $ n\geq n_0$ and define
\begin{equation}
p_1:=\max_{n<n_0}\overline{p_d(G_n)}\;.
\end{equation}
From Eq.~(\ref{eq:pbar}) we have  that $\overline{p_d(G_n)}\in[1/(d+1),1)$ for every $n$, which implies $p_1<1$ and
$\overline{p_d(G_n)}\leq\max\{p_0,p_1\}<1$ as claimed.

The implication (d) $\Rightarrow$ (e) follows immediately from the fact that for every sequence its $\limsup$ is bounded above by its supremum.

Next, we show that (e) $\Rightarrow$ (c). From the definition of the
$\limsup$ we have that
\begin{equation}
\inf_n\sup_{m\geq n}\overline{p_d(G_{m})}=p_0<1.
\end{equation}
Putting $p_1=(p_0+1)/2$ we have $p_0<p_1<1$
and, therefore, from the definition of infimum there exists a $n_0\in\mathbb{N}$ such that
\begin{equation}
\sup_{n\geq n_0}\overline{p_d(G_n)}<p_1.
\end{equation}
Hence, $\sigma_d(G_n,p_1)$ is GME for all $ n\geq n_0$ as claimed.

We conclude the proof by showing that (c) $\Rightarrow$ (b) $\Rightarrow$ (a).
The second implication is a direct consequence of
Proposition~\ref{prop:facts}~(ii). To show (c) $\Rightarrow$ (b) we reason as in the implication
(c) $\Rightarrow$ (d). Take now $p_0$ and $n_0$ to be defined as in the statement of (c) and let
\begin{equation}
p_1=\max_{n<n_0}\{\overline{p_d(G_n)}\}<1,\quad p_2=\frac{p_1+1}{2}<1.
\end{equation}
Therefore, for all $n$ we have
$\overline{p_d(G_n)}<\max\{p_0,p_2\}=p_3<1$ and so for all $n$
the state $\sigma_d(G_n,p_3)$ is GME.
\end{proof}

The next result gives a characterization of the ABS property for a sequence of graphs.

\begin{proposition}\label{characterizationABS}
The graph sequence $\cG=(G_n)_{n\in\N}$ is ABS in a given dimension $d$ iff \linebreak
$\lim_{n\to\infty}\overline{p_d(G_n)}=1.$
\end{proposition}
\begin{proof}
This is immediate given that $\overline{p_d(G_n)}\in[1/(d+1),1)$ for all $n$.
In fact, the sequence $\cG$ being ABS implies that
for any $\epsilon>0$ there exists $n_0\in\mathbb{N}$ such that for
all $n\geq n_0$ the state $\sigma_d(G_n,1-\epsilon)$ is BS.
By the hereditary property of biseparability stated in Proposition~\ref{prop:facts}~(ii), we conclude that being ABS is equivalent
to the following condition: for any $\epsilon>0$ there exists $n_0\in\mathbb{N}$ such that for all $n\geq n_0$ we have
$|\overline{p_d(G_n)}-1|\leq\epsilon$, which shows $\lim_{n\to\infty}\overline{p_d(G_n)}=1$ as claimed.
\end{proof}

\begin{remark}\label{rem:dico}
The notions of GME and BS are a dichotomy, i.e.\ a multipartite quantum state is either GME or BS. On the other hand, the results presented in this subsection
imply that if a sequence of graphs $\cG$ is ABS (resp. AGME), then it cannot be AGME (resp. ABS) but, in general, the reverse implication is not true. This proves that
the dichotomy GME/BS present at each stage disappears asymptotically. To see this consider the examples of graph sequences already presented in \cite{ASGME}: recall that
the sequence of complete graphs $(K_{n+1})_{n}$ is AGME while any sequence of trees $(T_{n+1})_n$ (where $T_n$ is a tree of order $n$) is ABS. It follows then that the graph sequence $\cG=(G_n)_{n\in\N}$ defined by $G_n=K_{n+1}$ if $n$ is odd and $G_n=T_{n+1}$ if $n$ is even is neither AGME nor ABS.
\end{remark}

The following characterizations will be useful in our proofs when we want to negate the AGME or ABS properties. Moreover, as a corollary, we obtain that all graph sequences that are neither AGME nor ABS are essentially of the artificial form mentioned in the preceding remark.

\begin{proposition}\label{prop:notAGMEnotABS}
Consider a sequence $\cG=(G_n)_{n\in\N}$. Then,
\begin{itemize}
\item[(i)] $\cG$ is not AGME in dimension $d$ iff $\cG$ has a subsequence which is ABS in dimension $d$.
\item[(ii)] $\cG$ is not ABS in dimension $d$ iff $\cG$ has a subsequence which is AGME in dimension $d$.
\end{itemize}
\end{proposition}
\begin{proof}
(i) For simplicity we denote
$\overline{p_n}:=\overline{p_d(G_n)}\in\left[{1}/{(d+1)},1\right)$. Then, by
Proposition~\ref{characterizationAGME}~(e), we have that $\cG$ is not AGME iff $\limsup_{n\to\infty}\overline{p_n}=1$
and this condition is equivalent to the existence of a subsequence
$\left(\overline{p_{s(n)}}\right)_n$ such that $\lim_n\overline{p_{s(n)}}=1$. Therefore, by Proposition~\ref{characterizationABS}, the result is proven.

(ii) To show the second characterization, note that, by Proposition~\ref{characterizationABS}, $\cG$ is not ABS iff
$\lim\overline{p_n}\not=1$ and this condition holds iff
\begin{equation}
 \quadtext{either}
 \limsup_{n\to\infty}\overline{p_n}<1 \quadtext{or}
 \limsup_{n\to\infty}\overline{p_n}=1\;,\;\mathrm{but}\;
 \left(\overline{p_n}\right)_n\;\;\mathrm{is~not~convergent}\,.
\end{equation}
In both cases, Proposition~\ref{characterizationAGME}~(e) implies that either
$\cG$ is already AGME or there exists a subsequence $\left(\overline{p_{s(n)}}
\right)_n$ with $\lim_n\overline{p_{s(n)}}<1$ and the corresponding subsequence of graphs $\cG_s=(G_{s(n)})_n$ is AGME.
\end{proof}

\begin{corollary}
A graph sequence $\cG=(G_n)_n$ is neither AGME nor ABS in dimension $d$ iff it contains both a subsequence which is ABS in dimension $d$ and a subsequence which is AGME in dimension $d$.
\end{corollary}

\begin{remark}
We conclude the section with a few comments on the monotonic behavior of the dimension $d$ for these asymptotic notions. If the graph sequence $\cG=(G_n)_n$ is ABS in dimension $d$, then it is also ABS in any dimension $d'<d$. This follows from
$\overline{p_{d'}(G)}\geq\overline{p_d(G)}$ for all $d'\leq d$
since local projections transform $\rho(p,d)$ into $\rho(p,d')$ and, hence, biseparability of the former implies biseparabilty of the latter.
Moreover, if $\cG$ is AGME in dimension $d$, then it is also AGME in any dimension $d'>d$.
This is a consequence of the above and Proposition~\ref{prop:notAGMEnotABS}~(i). Thus, in particular, if a graph sequence is AGME when $d=2$, then it is AGME
(recall Definition~\ref{def:asymptotic}~(i)).
\end{remark}

\section{A sufficient condition for AGME}\label{sec3}

In this section we prove a graph theoretic condition that
guarantees that a sequence of graphs is AGME (see Proposition~\ref{prop:sufficientASGME}). This will be a very useful tool in order to establish our main results in the sections that follow.

Recall that in the multipartite scenario the Hilbert space is given by $H=\bigotimes_{i=1}^NH_i$ and for any $i,j\in[N]$ ($i\neq j$) let
\begin{equation}\label{loccij}
\Lambda_{ij}\colon D\left(\bigotimes_{k=1}^{N}H_{k}\right)\rightarrow D(H_{i}\otimes H_{j})
\end{equation}
be an arbitrary quantum channel that maps $N$-partite states to bipartite states shared between parties $i$ and $j$. We denote
a collection of such channels for each pair of parties simply by
$\mathbf{\Lambda}=\{\Lambda_{ij}\}_{i<j}$. If in addition $\phi_{ij}^{+}\in D(H_i\otimes H_j)$ is the 2-dimensional maximally entangled state shared by parties $i$ and $j$, we define the average maximally-entangled-state fidelity of $\mathbf{\Lambda}$ acting on $\rho\in D(H)$ by
\begin{equation}\label{fidelitylambda}
F_\mathbf{\Lambda}(\rho):=\frac{2}{N(N-1)}\sum_{i<j}
F\left(\Lambda_{ij}(\rho),\phi_{ij}^{+}\right).
\end{equation}
The following sufficient condition for biseparability was established in \cite[Lemma 1]{ASGME}.
\begin{lemma}\label{lemmabsfidelity}
If $\rho\in \mathrm{BS}(H)$ (i.e.\ it is an $N$-partite BS state) and $\mathbf{\Lambda}$ is any collection of LOCC channels of the form given in
Eq.~(\ref{loccij}), then
\begin{equation}\label{eq:fidelityBSlemma}
F_\mathbf{\Lambda}(\rho)\leq1-\frac{1}{N}\;.
\end{equation}
\end{lemma}
This condition was used in \cite{ASGME} to prove that the sequence of
complete graphs $(K_{n+1})_n$ is AGME. In particular, it is shown that there exists a collection of LOCC protocols based on sequential teleportation and entanglement distillation with the property that for every given $d$ there exists a $p_0<1$ such that for all $p\geq p_0$ we have
\begin{equation}\label{lemma2}
F_\mathbf{\Lambda}(\sigma_d(K_{n+1},p))\geq1-\epsilon_n
\end{equation}
with $\epsilon_n\to 0$ exponentially fast as $n\to\infty$, hence contradicting Eq.\ (\ref{eq:fidelityBSlemma}) for $n$ large enough. Here, we put these ideas in a much more general context, which makes it possible to exploit
Lemma~\ref{lemmabsfidelity} so as to prove AGME not only for $(K_{n+1})_n$ but for arbitrary graph sequences $(G_n)_n$ fulfilling some purely graph-theoretic condition.

\begin{proposition}\label{prop:sufficientASGME}
Let $\cG=(G_n)_n$ be a graph sequence with the property that for every pair of vertices in $G_n$ there exist $f(n)$ edge-disjoint paths that connect them with length at most $C$, where $C$ is a universal constant and $f(n)\in\omega (\log |G_n|)$. Then,
$\cG$ is AGME.
\end{proposition}
\begin{proof}
Let us fix the dimension $d$. As it is explained in \cite[Lemma 2]{ASGME}, the results of \cite{hamada,devetakwinter} guarantee that there exists a value $q=q(d)\in (0,1)$ such that for every $n\geq 1$ and every $p\geq q$,
we can find an LOCC channel
\begin{equation}\Lambda_n(p)\colon D\left(\C^{d^n}\otimes \C^{d^n}\right)
\rightarrow D(\C^2\otimes \C^2)
\end{equation}
satisfying
\begin{equation}\label{eq:hamada}
F\big(\Lambda_n(p)(\rho(p,d)^{\otimes n}),\phi^+\big)\geq 1-\epsilon_n\;,
\quadtext{with}
\epsilon_n\leq2(n+1)^{2(d^2-1)}d^{-Rn}\;,
\end{equation}
where $R=R(p,d)>0$ for every $p\geq q$ and $d$.
Now, let $e_1=\{i,k\}$ and
$e_2=\{k,j\}$ be two edges that connect vertices $i$ and $j$ in a path of length 2. Then, by standard teleportation, there is an LOCC channel that transforms $\rho_{e_1}(p,d)\otimes(\phi^+_d)_{e_2}$ into $\rho_{ij}(p,d)$.
However, in our case the edge $e_2$ does not support a perfect maximally entangled state but an isotropic state with visibility $p$, which must then be used as a resource for noisy teleportation. A simple calculation (cf.\ again Lemma 2 in \cite{ASGME}) shows that executing the standard teleportation protocol with this noisy resource state transforms $\rho_{e_1}(p,d)\otimes\rho_{e_2}(p,d)$ into $\rho_{ij}(p^2,d)$. Suppose then that parties $i$ and $j$ in $G$ are connected by a path of length $L$ through the edges in $F\subset E$. By sequentially teleporting the isotropic state corresponding to the first edge in the path at the end of the process parties $i$ and $j$ end up sharing a noisier version of the isotropic state, i.e.\ there is an LOCC protocol that transforms $\bigotimes_{e\in F}\rho_e(p,d)$ into $\rho_{ij}(p^{2^{L-1}},d)$. Thus, if there are $f(n)$ edge-disjoint paths in $G_n$ connecting an arbitrarily chosen pair $i,j$, then, by sequential teleportation as above and discarding all edges that do not belong to any of these paths, there is an LOCC channel which transform the state  $\sigma_d(G_n,p)$ to
\begin{equation}
    \bigotimes_{k=1}^{f(n)}\rho_{ij}\left(p^{2^{L_k(n)-1}},d\right),
\end{equation}
where $L_k(n)$ is the length of each edge-disjoint path. However, as pointed out in Section~\ref{sec21}, the visibility of an isotropic state can always be decreased by further LOCC processing. Therefore, there is an LOCC channel $\Lambda^n_{ij}$ such that
\begin{equation}
    \Lambda^n_{ij}(\sigma_d(G_n,p))=\rho_{ij}(p^{2^{C-1}},d)^{\otimes f(n)}\;,
\end{equation}
where $C$ is the universal constant such that $L(n)=\max_kL_k(n)\leq C$ for every $n$ required in the statement of the proposition. Moreover, if we let $p_0=q^{1/2^{C-1}}\in (0,1)$, for every $p\geq p_0$, we can replace the previous LOCC channel $\Lambda^n_{ij}$ by another one $\Lambda^n_{ij}(p)$ such that
\begin{equation}
    \Lambda^n_{ij}(p)(\sigma_d(G_n,p))=\rho_{ij}(p_0^{2^{C-1}},d)^{\otimes f(n)}.
\end{equation}

Hence, we have that for every $p\geq p_0$ and for every  $n\geq 1$ there exists an LOCC channel $\tilde{\Lambda}^{i,j}_n(p)=\Lambda_n(p)\circ \Lambda^n_{ij}(p)$ such that
\begin{align}
F\left(\tilde{\Lambda}^{i,j}_n(p)(\sigma_d(G_n,p)) ,\phi^+\right)=F\left(\Lambda_n(p)(\rho_{ij}(q,d)^{\otimes f(n)}) ,\phi^+\right)\geq 1-\epsilon_n,
\end{align}where $\epsilon_n\leq2(f(n)+1)^{2(d^2-1)}d^{-R(q,d)f(n)}$.

Note also that, by our assumptions on the function $f(n)$, there exists $n_0\in\mathbb N$ such that for every $n\geq n_0$, we have that $|G_n|\epsilon_n<1$. Therefore, for every $p\geq p_0$, every  $n\geq n_0$ and every pair $(i,j)$ with $i<j$ we can find an LOCC channel $\tilde{\Lambda}^{i,j}_n(p)$ such that
\begin{align}
F\left(\tilde{\Lambda}^{i,j}_n(p)(\sigma_d(G_n,p)) ,\phi^+\right)> 1-\frac{1}{|G_n|}.
\end{align}

According to Lemma~\ref{lemmabsfidelity}, we conclude that $\sigma_d(G_n,p)$ is GME for every $p\geq p_0$ and $n\geq n_0$. So, item c) in
Proposition~\ref{characterizationAGME} tells us that $\cG$ is AGME.
\end{proof}

\begin{remark}\label{Remark:sufficientASGME}
In fact, the assumption made in the previous proposition on the function $f$ guarantees that if $\epsilon_n\leq2(f(n)+1)^{2(d^2-1)}d^{-R(q,d)f(n)}$, then $\lim_n |G_n|\epsilon_n=0$. Therefore, it can be deduced from the previous proof that for every fixed $d\geq 2$, there exists a value $p_0=p_0(d)\in (0,1)$ with the following property: for all $p\geq p_0$ and $n\geq 1$ there exists a collection of LOCC protocols
$\mathbf{\Lambda}^n(p)=\{\Lambda_{ij}^n(p)\}_{i<j}$ satisfying that
\begin{equation}
F_\mathbf{\mathbf{\Lambda}^n(p)}(\sigma_d(G_n,p)):=\frac{2}{|G_n|(|G_n|-1)}\sum_{i<j}
F\left(\Lambda_{ij}^n(p)(\sigma_d(G_n,p)),\phi_{ij}^{+}\right)> 1-\epsilon_n,
\end{equation}with $\lim_n |G_n|\epsilon_n=0$.
\end{remark}

It is worth pointing out that the sufficient condition for AGME of Proposition~\ref{prop:sufficientASGME} is not necessary. Later on, we will provide constructions of graph sequences that do not obey it, which we show anyway to be AGME. Interestingly, Proposition~\ref{prop:sufficientASGME} not only makes it possible to prove AGME for graph sequences meeting its premises. We will later see how to apply it for other graph sequences for which this is not the case but which can be partitioned in a certain way into subgraphs with the desired property. In order to achieve this, we will make use of the next lemma, which shows that the graph sequences fulfilling the conditions of
Proposition~\ref{prop:sufficientASGME} satisfy in fact something much stronger than AGME.

\begin{lemma}\label{lemmaBSAto0}
Let $\cG=(G_n)_n$ be a graph sequence that satisfies the hypotheses of
Proposition~\ref{prop:sufficientASGME}. Then, for every given $d$, there exists a value $p_0<1$ such that if $p\geq p_0$, for any convex decomposition
\begin{equation}\label{eqlemmaBSAto0}
\sigma_d(G_n,p)=q_n\rho_n+(1-q_n)\tau_n\;,
\end{equation}
where $\tau_n$ and $\rho_n$ are $|G_n|$-partite quantum states\footnote{
We do not write the dependence on $p$ (in both the coefficient $q_n$ and the states $\rho_n$ and $\tau_n$) to simplify the notation.} and the latter is moreover biseparable, it holds that $\lim_{n\to\infty} q_n=0$.
\end{lemma}
\begin{proof}
According to Remark~\ref{Remark:sufficientASGME}, for every fixed $d\geq 2$, there exists a value $p_0\in (0,1)$ such that for all $p\geq p_0$ and $n\geq 1$ there exists a collection of LOCC protocols $\mathbf{\Lambda}^n(p)=\{\Lambda_{ij}^n(p)\}_{i<j}$ satisfying that
\begin{equation}\label{fidelitylambda_remark}
F_{\mathbf{\Lambda}^n(p)}(\sigma_d(G_n,p)):=\frac{2}{|G_n|(|G_n|-1)}\sum_{i<j}
F\Big(\Lambda_{ij}^n(p)(\sigma_d(G_n,p)),\phi_{ij}^{+}\Big)> 1-\epsilon_n,
\end{equation}
with $\lim_n |G_n|\epsilon_n=0$

Now, the linearity of
$F(\cdot,\cdot)$ (in both arguments) when one argument is a pure state implies that for any given set of channels $\mathbf{\Lambda}$, $F_\mathbf{\Lambda}(\rho)$ is linear. Thus, if $\mathbf{\Lambda}^n(p)$ is the collection of channels specified above, by Eq.\ (\ref{eqlemmaBSAto0}) we have that
\begin{equation}
1- \epsilon_n< F_\mathbf{\mathbf{\Lambda}^n(p)}(\sigma_d(G_n,p))\leq1-\frac{q_n}{|G_n|},
\end{equation}
where we have used that $F_\mathbf{\mathbf{\Lambda}^n(p)}(\tau_n)\leq1$ and, by Lemma~\ref{lemmabsfidelity}, $F_\mathbf{\mathbf{\Lambda}^n(p)}(\rho_n)\leq1-1/|G_n|$.

The previous chain of inequalities implies that for $p\geq p_0$, we have
$q_n< |G_n|\epsilon_n$, from which the result follows.
\end{proof}

\section{Edge-connectivity and AGME}\label{seccon}

The intuition that the more connected a graph sequence is, the more likely it is to be AGME can be formalized in terms of edge deletion. Let $G$ be a connected graph and $G'$ a connected spanning subgraph of $G$ (recall Section~\ref{subsec:GT}).
Since $E'=E\setminus F$ for some suitable subset of edges $F\subset E$ we have (recall Eq.~(\ref{edgedeletion}))
\begin{equation}
\sigma_d(G',p)=\tr_{\left(\bigotimes_{e\in F}H_e\right)}\sigma_d(G,p).
\end{equation}
The key observation here is that edge deletion preserves biseparability. In fact, from
the definition of biseparability, it is clear that partial traces that preserve the number of parties must map BS states into BS states. Therefore, if $\sigma_d(G,p)$ is BS, then $\sigma_d(G',p)$ must be BS as well and the critical visibilities must satisfy that $\overline{p_d(G)}\leq\overline{p_d(G')}$. Similarly, if $\cG'=(G'_n)_n$ is a sequence of connected spanning subgraphs of $\cG=(G_n)_n$ by the results of Subsection~\ref{sec22} it follows that if $\cG'$ is AGME then so must be $\cG$, and that if $\cG$ is ABS then the same holds for $\cG'$. This observation gives a natural perspective to the results of \cite{ASGME}.
Since every graph can be obtained by edge deletion from $K_n$ for some $n$, if an AGME graph sequence $\cG=(G_n)_n$ with $|G_n|=n+1$ exists, $(K_{n+1})_n$ must be such a sequence. Analogously, given that every graph contains a spanning tree and trees become disconnected by the removal of any edge, there has to be a tree sequence that is ABS if such a sequence exists.

In this section we study how the asymptotic growth of the edge-connectivity
of the underlying network intertwines with the AGME and ABS properties to see if the above intuition can be made more general in rigorous quantitative terms. Proposition~\ref{th:lambda} shows that indeed large edge-connectivity leads to AGME. However, this illuminates only partially our problem.
In the proof of Proposition~\ref{th:lowlambda} (see also Fig.~\ref{fig:mancuerna})
we construct a graph sequence with minimal edge connectivity that is AGME and, consequently, the condition of Proposition~\ref{th:lambda} cannot be necessary. Furthermore, this example shows that no asymptotic upper bound on the edge-connectivity can serve as a sufficient condition for the ABS property. In addition to this, the results of this section are used later on in order to derive stronger conditions in terms of the degree growth.

\subsection{Large edge-connectivity implies AGME}

We begin recalling
the following estimate on the diameter of a graph in terms of its order and minimal degree
(see \cite[Theorem~1]{erdos}).\footnote{The mentioned theorem in \cite{erdos}
proves the first estimate; we include the second trivial estimate since it is enough for our purposes and simplifies expressions that will appear later.}
\begin{theorem}\label{thm:erdos}
For a connected graph $G$ with $\delta_{\mathrm{min}}(G)\geq 2$ one has
\begin{equation}\label{ineqdiam}
\mathrm{diam}(G)\leq\frac{3|G|}{\delta_{\mathrm{min}}(G)+1}-1
\leq \frac{3|G|}{\delta_{\mathrm{min}}(G)}.
\end{equation}
\end{theorem}

\begin{proposition}\label{th:lambda}
If $\cG=(G_n)_{n\in\N}$ is a graph sequence such that $\lambda(G_n)\in\Omega(|G_n|)$,
then $\cG$ is AGME.
\end{proposition}
\begin{proof}
By hypothesis, we know that there exists a constant $c>0$ and a natural number $n_0$ such that for every $n\geq n_0$ we have $\lambda(G_n)\geq c|G_n|$ and we can assume that $c\in (0,1)$. Now, given $n\geq n_0$ and $u, v\in V$ with $u\neq v$, we will show that there are at least $\lfloor c^2|G_n|/6\rfloor$ edge-disjoint paths between the vertices $u$ and $v$  with length at most $9/(3c-c^2)$. In this case, we can apply Proposition~\ref{prop:sufficientASGME} to conclude the proof.

To show the estimates on the edge-disjoint paths announced in the previous paragraph, let us fix $n\geq n_0$ and consider the following sequence of graphs defined recursively and obtained from $\cG$ by successive edge deletions: $G_{n,0}:=G_n$ and for the next step consider a shortest path between $u$
and $v$ (which exists since $G_n$ is connected). The graph $G_{n,1}$ is constructed from $G_{n,0}$ by erasing all edges corresponding to this shortest path. Iterating this process until there are no paths connecting the pair of vertices we obtain the family $(G_{n,k})_k$. We are going to show that $G_{n,k}$ is connected and $\textrm{diam}(G_{n,k})\leq 9/(3c-c^2)$ for every $k=0,\ldots , \lfloor c^2|G_n|/6\rfloor-1$. To do so, we will use mathematical induction (on a finite set). More precisely, we will show:
\begin{itemize}
\item[1.] $G_{n,0}$ is connected and $\textrm{diam}(G_{n,0})\leq 9/(3c-c^2)$.
\item[2.] If $G_{n,k}$ is connected, $\textrm{diam}(G_{n,k})\leq 9/(3c-c^2)$ and $k\leq \lfloor c^2|G_n|/6\rfloor-2$, then $G_{n,k+1}$ is connected and $\textrm{diam}(G_{n,k+1})\leq 9/(3c-c^2)$.
\end{itemize}

First, $G_{n,0}$ is connected by definition and, according to
Theorem~\ref{thm:erdos} and recalling that \linebreak $\delta_{\mathrm{min}}(G_n)\geq \lambda(G_n)\geq c|G_n|$, we have
\begin{equation}
\textrm{diam}(G_{n,0})=\textrm{diam}(G_{n})\leq \frac{3|G_n|}{\delta_{\mathrm{min}}(G_n)}\leq \frac{3}{c}\leq \frac{9}{3c-c^2}.
\end{equation}

Let us now assume that $G_{n,k}$ is connected, $\textrm{diam}(G_{n,k})\leq 9/(3c-c^2)$ and $k\leq \lfloor c^2|G_n|/6\rfloor-2$. The connectivity of $G_{n,k}$ guarantees the existence of $G_{n,k+1}$. Note also that our assumption on the diameter implies that $\textrm{diam}(G_{n,s})\leq 9/(3c-c^2)$ for every $s=0,1,\ldots, k$. Hence, the total number of edges deleted from $G_n$ in order to build $G_{n,k+1}$ is not larger than
\begin{equation}
\frac{9(k+1)}{(3-c)c}\leq \frac{9\left(c^2|G_n|/6-1\right)}{(3-c)c}\leq  \frac{\frac{3}{2} c|G_n|}{(3-c)}\leq \frac{3}{4}c|G_n|<\lambda(G_n),
\end{equation}where we have used that $c\in (0,1)$. This implies that $G_{n,k+1}$ is connected.

To obtain an upper bound on $\textrm{diam}(G_{n,k+1})$, first note that
\begin{equation}
\delta_{\mathrm{min}}(G_{n,k+1})\geq\delta_{\mathrm{min}}(G_{n,k})-2\geq\ldots\geq \delta_{\mathrm{min}}(G_{n})-2(k+1)\geq c|G_n|-2(k+1)\,,
\end{equation}
where we have used that deleting a path decreases the degree of a vertex at most by $2$.

The fact, that $k\leq c^2|G_n|/6-2$, guarantees that $\delta_{\mathrm{min}}(G_{n,k+1})\geq  c|G_n|-2(k+1)\geq 2$, so we can apply Theorem~\ref{thm:erdos} to obtain

\begin{equation}
\textrm{diam}(G_{n,k+1})\leq\frac{3|G_{n,k+1}|}{\delta_{\mathrm{min}}(G_{n,k+1})}=\frac{3|G_{n}|}{\delta_{\mathrm{min}}(G_{n,k+1})}\leq\frac{3|G_n|}{c|G_n|-2(k+1)}.
\end{equation}

Moreover, our upper bound for $k$ implies $c|G_n|-2(k+1)\geq c|G_n|-c^2|G_n|/3$, from where one immediately obtains the claimed upper bound
\begin{equation}
\textrm{diam}(G_{n,k+1})\leq \frac{9}{3c-c^2}
\end{equation}
and the proof is concluded.
\end{proof}

Recall from Subsection~\ref{subsec:GT} that the vertex and edge connectivity of any graph $G$ satisfy that  $\lambda(G)\geq \kappa(G)$. Hence, we obtain directly the following result.
\begin{corollary}\label{th:kappa}
If $\cG=(G_n)_{n\in\N}$ is a graph sequence such that
$\kappa(G_n)\in\Omega(|G_n|)$, then $\cG$ is AGME.
\end{corollary}

\subsection{AGME with minimal connectivity}

We present next an explicit construction of a minimally connected graph sequence displaying AGME.

\begin{figure}[h]
 \centering
\begin{tikzpicture}
    \node[draw, circle, inner sep=1pt, minimum size=4pt, fill=black] (a) at (0, 0) {};  
    \draw[black] (0, 0) ellipse (0.5cm and 0.75cm);  
    \node at (0, 0.3) {$K_n$};  

    \node[draw, circle, inner sep=1pt, minimum size=4pt, fill=black] (b) at (4, 0) {};  
    \draw[black] (4, 0) ellipse (0.5cm and 0.75cm);  
    \node at (4, 0.3) {$K_n$};  

    \draw (a) -- (b);

    \node at (2, -0.2) {$e$};
\end{tikzpicture}
 \caption{Graph $G_n$ of a minimally edge connected AGME graph sequence $\cG=(G_n)_n$. Circles represent clusters, which in this case are completely connected.}
\label{fig:mancuerna}
\end{figure}

\begin{proposition}\label{th:lowlambda}
There exists an AGME graph sequence $\cG=(G_n)_{n\in\N}$ such that
$\lambda(G_n)=\kappa(G_n)=1$ for all $n$.
\end{proposition}
\begin{proof}
To prove the statement we will construct a concrete graph sequence $\cG=(G_n)_{n\in\N}$ that is AGME and has minimal edge-connectivity (see Fig.~\ref{fig:mancuerna}).
For each $n\in\N$ the graph $G_n=(V_n,E_n)$ has order $|G_n|=2n$ and we partition the vertex set into two clusters of order $n$ by
$V_n=V_n^{(1)}\sqcup V_n^{(2)}$ which we also numerate
as $V_n^{(1)}=[n]$ and $V_n^{(2)}=\overline{[n]}=\{n+1,\dots,2n\}$. All vertices
in each cluster $V_n^{(i)}$, $i=1,2$, are pairwise connected (i.e., the corresponding induced subgraphs are complete graphs $K_n$). Moreover, the two clusters are connected by a unique bridge-edge $\{n,n+1\}\in E_n$. Thus, the set of edges of the graph can be partitioned as
$E_n=E_n^{(1)}\sqcup E_n^{(2)}\sqcup E_n^{(3)}$ with
\begin{equation}
    E_n^{(1)}=\{\{i,j\}\mid 1\leq i<j\leq n\}\,,\, E_n^{(2)}=\{\{i,j\} \mid n+1\leq i<j\leq 2n\}\,,\, E_n^{(3)}=\{\{n,n+1\}\}.
\end{equation}
Notice that, by construction, $\lambda(G_n)=\kappa(G_n)=1$ for any $n$,
since $G_n$ has one bridge-edge.

Next, we will show that $\cG$ is AGME by contradiction. Assume then that $\cG$ is not AGME
in a given dimension $d$.
Then, Proposition~\ref{prop:notAGMEnotABS}~(i) implies that there is a subsequence, $\cG_s=(G_{s(n)})_n$, which is ABS in dimension $d$. Let $p_0\in (0,1)$ be the value guaranteed by Lemma~\ref{lemmaBSAto0} when applied to the complete graph  $K_{s(n)}$ (which obviously fulfills the hypothesis of Proposition~\ref{prop:sufficientASGME} since it has $s(n)-1$ edge-disjoint paths with length $2$). The ABS nature of the subsequence $\cG_s$ implies that for every $p>\max\{p_0, 1/(d+1)\}$ there exists $n_0=n_0(d,p)$ such that for all $n\geq n_0$ we have
\begin{equation}\label{eqprop1}
\sigma_d(G_{s(n)},p)
    =\sum_{\emptyset\neq M\subsetneq[2n]}q_M^{(n)}\chi_M^{(n)}
    \quadtext{with} q_M^{(n)}\geq0\quadtext{and}
    \sum_Mq_M^{(n)}=1\;,
\end{equation}and where $\chi_M^{(n)}$ is an $2n$-partite quantum state that is separable in the bipartition $M|\overline{M}$ (note that these objects depend on the choice of $d$ and $p$ but we omit for simplicity to write them explicitly).

In the following it will be useful to divide the subsets $M$ over which the above sum runs into four different classes $\{I_i^n\}_{i=1}^4$. The first one,
$I_1^n$, will contain all subsets such that either $M\subsetneq[n]$ or $\overline{M}\subsetneq[n]$ (thus, it represents all bipartitions that only split the first half of the parties). Analogously, $I_2^n$ considers the bipartitions that only split the second half of the parties: it will contain all subsets such that either $M\subsetneq\overline{[n]}$ or $\overline{M}\subsetneq\overline{[n]}$ and $I_3^n$ corresponds to a single bipartition $I_3^n=\{[n],\overline{[n]}\}$. Lastly, $I_4^n$ contains all remaining subsets (that is, those corresponding to bipartitions that split both the first and the second cluster of the parties). Notice now that
\begin{equation}
    \tr_{\left(\bigotimes_{e\in \left(E_n^{(2)}\sqcup E_n^{(3)}\right)}H_e\right)}\sigma_d(G_{s(n)},p)=\sigma_d(K_{s(n)},p).
\end{equation}
Applying this partial trace to the right-hand-side of Eq.~(\ref{eqprop1}), we obtain a decomposition as the one given in Eq.~(\ref{eqlemmaBSAto0})
\begin{equation}
\sigma_d(K_{s(n)},p))=q_n\rho_n+(1-q_n)\tau_n
\quadtext{with}
q_n=\sum_{M\in I_1^n,I_4^n}q_M^{(n)}\,.
\end{equation}
Now, since we have chosen $p>p_0$, Lemma~\ref{lemmaBSAto0} implies that
$\lim_{n\to\infty}\sum_{M\in I_1^n,I_4^n}q_M^{(n)}=0$.

Similarly, taking the partial trace with respect to the subspace
$\bigotimes_{e\in E_n^{(1)}\sqcup E_n^{(3)}}H_e$ we obtain now
$\lim_{n\to\infty}\sum_{M\in I_2^n,I_4^n}q_M^{(n)}=0$. In particular, we obtain
\begin{equation}\label{eq:I}
 \lim_{n\to\infty}
  \left(
  \sum_{M\in I_1^n}q_M^{(n)}+\sum_{M\in I_2^n}q_M^{(n)}+\sum_{M\in I_4^n}q_M^{(n)}
  \right)=0\;.
\end{equation}

Next, if we take now the partial trace with respect to $\bigotimes_{e\in E_n^{(1)}\cup E_n^{(2)}}H_e$ in both sides of Eq.~(\ref{eqprop1}) we obtain
\begin{equation}
\rho(p,d)=q_{[n]}^{(n)}\tilde{\rho}_n+\left(1-q_{[n]}^{(n)}\right)\tilde{\tau}_n\;,
\end{equation}
where $\tilde\tau_n$ and $\tilde\rho_n$ are 2-qudit quantum states and the latter is moreover separable (note that, without loss of generality, we can assume for all $n$ that $q_{\overline{[n]}}^{(n)}=0$.)
Then, Eqs.~(\ref{bsa2}) and (\ref{bsaisotropic}) imply that for all $n\geq n_0$ we have
\begin{equation}\label{eq:M0}
q_{[n]}^{(n)}\leq \mathrm{BSA}(\rho(p,d))=\frac{d+1}{d}(1-p)
<1,
\end{equation}where in the last inequality we have used that $p>1/(d+1)$.

Finally, since $I_1^n\sqcup I_2^n\sqcup I_3^n\sqcup I_4^n$ gives the power set of $[2n]$ (excluding $\emptyset$ and $[2n]$) we must have on the one hand the normalization
\begin{equation}\label{eq:contra}
 1=\sum_M q^{(n)}_M
  = q_{[n]}^{(n)}+\sum_{M\in I_1^n}q_M^{(n)}
    +\sum_{M\in I_2^n}q_M^{(n)}+\sum_{M\in I_4^n}q_M^{(n)}\,
\end{equation}for every $n$ but, on the other hand, Eq.~(\ref{eq:I}) and the inequality in (\ref{eq:M0})
show that the right hand side of Eq.~(\ref{eq:contra}) is strictly smaller than $1$
for $n$ large enough, leading to a contradiction.
\end{proof}

The construction of the sequence of graphs in the preceding proof can be easily generalized to other sequences of connected graphs with minimal edge-connectivity. In addition to this, notice that, due to Menger's theorem (cf. Section~\ref{subsec:GT}), this graph sequence does not satisfy the hypothesis of
Proposition~\ref{prop:sufficientASGME}. Thus, this shows that the latter condition is sufficient for AGME but not necessary.

\section{Conditions based on the degree growth}\label{secdeg}

In the previous section we have seen that there is indeed a precise mathematical relation between connectivity and AGME. However, connectivity alone cannot characterize AGME and it cannot say anything about the ABS property. This is because graph sequences with the smallest possible edge-connectivity can nevertheless be AGME. In this section we turn our attention to other graph parameters. A natural choice to consider is the sequence of maximal/minimal degrees associated to $\cG$ (cf.\ Eq.~(\ref{maxmindegree})).
These numbers are relevant graph-theoretic parameters that contain information related to the connectivity and have also an important physical meaning in our context. They quantify the maximal/minimal local dimension of the underlying isotropic network state and are, therefore, directly related to the overall cost of preparing and manipulating the network.

As we will see in this section, the asymptotic behavior of the maximal/minimal degree sequences provide
sufficient conditions that guarantee either ABS (\textit{slow} degree growth, see Theorem~\ref{th:degABS}) or
AGME (\textit{fast} degree growth, see Theorem~\ref{th:degAGME}), the
latter condition giving a stronger result than Proposition~\ref{th:lambda}.

\subsection{Slow degree growth implies the ABS property}\label{subsecdeg1}

\begin{theorem}\label{th:degABS}
Any graph sequence $\cG=(G_n)_{n\in\N}$ for which the associated sequence of maximal degrees satisfies
$\delta_{\mathrm{max}}(G_n)\in o\left(\log|G_n|\right)$ is ABS.
\end{theorem}

\begin{proof}
For any given $d$, which is considered fixed in the following, we will simplify the notation of the isotropic
states in Eq.~(\ref{isotropicd}) putting $\phi^+$ and $\none$ instead of $\phi^+_d$ and $\none_{d^2}$, respectively, and use edge subindices to label the Hilbert spaces on which the operators act.

According to Definition~\ref{def:asymptotic}~(ii) we need to show that the isotropic network state
$\sigma_d(G_n,p)$ given in Eq.~(\ref{networkGdp}) is biseparable for all $p<1$ and for $n$ large enough.
Although it is not strictly necessary, we will show in order to simplify the reasoning that
$\sigma_d(G_n,p)$ is actually 1-biseparable. For this, rewrite first (using distributivity) the isotropic network state in terms of $2^{|E_n|}$ summands labeled by subsets of the edge set $E_n$ as
\begin{equation}\label{eq:summands}
\sigma_d(G_n,p)=\sum_{F_n\subseteq
E_n}p^{|F_n|}(1-p)^{|E_n|-|F_n|}\bigotimes_{e\in F_n}\phi^+_e\bigotimes_{e\notin
F_n}\none_e\;.
\end{equation}
We will split the preceding sum into three parts labeled by families of edge subsets
corresponding to different configurations of connectivity of vertices.
For this denote for any subset $F_n\subset E_n$
the edges in $F_n$ adjacent with a vertex $v\in V_n$ by
\begin{equation}
 F_n(v):=F_n\cap E_n(v)\;.
\end{equation}

\begin{itemize}
 \item We consider first all subsets of edges that label in the previous sum
\textit{non} 1-biseparable states:
\begin{equation}\label{bn}
 B_n=\{F_n\subseteq E_n \mid
 F_n(v)\neq\emptyset\quadtext{for any} v\in V_n\}.
\end{equation}

\item Next, we consider all subsets of edges labeling states that
are 1-biseparable \textit{exactly} at one vertex $v$:
\begin{equation}\label{an}
A_n=\{F_n'\subseteq E_n \mid \textrm{ there exists a unique }
      v\in V_n \textrm{ such that } F_n'(v)=\emptyset\}\;.
\end{equation}
In fact, if $F_n'\in A_n$ and for (exactly) one vertex $v\in V_n$
we have that $F_n'(v)=\emptyset$, then (see Subsection~\ref{sec21})
\begin{equation}
\bigotimes_{e\in {F'_n}}
\phi^+_e\bigotimes_{e\notin {F'_n}}\none_e\in S_v(H).
\end{equation}
Thus, mixtures of these terms for different
choices of $v$ give rise to
a 1-biseparable state.

\item Finally, we consider all remaining subsets of edges,
\begin{equation}\label{cn}
C_n= E_n\setminus \left(A_n\cup B_n\right)\;,
\end{equation}
which manifestly correspond to 1-biseparable terms.
\end{itemize}

Note that for any $F_n\in B_n$ and for any $v\in V_n$ there exists a
unique $F'_n\in A_n$ such that $F'_n(v)=\emptyset$ and $F_n(w)=F'_n(w)$ for all $w\neq v$ ($F'_n$ is obtained from $F_n$ by removing all edges adjacent to $v$).
The idea of the proof consists in recombining non-1-biseparable terms in $B_n$
with the 1-biseparable terms from $A_n$ (with appropriate weights) so
as to overall obtain a 1-biseparable state for $n$ sufficiently large. In this
process the summands labeled by $C_n$ will remain untouched as they already correspond to
1-biseparable states.
To establish the weights needed in the mixture notice that any $F_n\in B_n$ can be
combined with $|G_n|$ terms in $A_n$ with the preceding property (one for each choice of vertex). However, if we want to do so
for all terms in $B_n$, each term $F'_n\in A_n$ such that $F'_n(v)=\emptyset$
for a given $v\in V_n$ needs to be used for all terms $F_n\in B_n$ such that
$F_n(w)=F'_n(w)$ for all $w\neq v$ and note that there are $2^{\delta_n(v)}-1$
such terms (one for each possible choice of $F_n(v)$ excluding the empty
set).

Considering the vertex refined decomposition of the $B_n$ summands as
\begin{eqnarray}
\lefteqn{\sum_{F_n\in B_n}
p^{|F_n|}(1-p)^{|E_n|-|F_n|}\bigotimes_{e\in
F_n}\phi^+_e\bigotimes_{e\notin F_n}\none_e} \qquad\qquad\nonumber\\[3mm]
&=&
\sum_{F_n\in B_n}\sum_{v\in V_n}\frac{1}{|G_n|}\;\;
p^{|F_n|}(1-p)^{|E_n|-|F_n|}\bigotimes_{e\in
F_n}\phi^+_e\bigotimes_{e\notin F_n}\none_e\;
\end{eqnarray}
(and simmilarly for the $A_n$ summands with weights
$\left(2^{\delta_n(v)}-1\right)^{-1}$) we can rewrite Eq.~(\ref{eq:summands}) as
\begin{eqnarray}
\sigma_d(G_n,p)&=&
\sum_{F_n\in B_n}
p^{|F_n|}(1-p)^{|E_n|-|F_n|}\bigotimes_{e\in
F_n}\phi^+_e\bigotimes_{e\notin F_n}\none_e
+ \sum_{F_n'\in A_n}\Big(\cdots \Big)
+ \sum_{F_n\in C_n}\Big(\cdots \Big)
\nonumber \\[3mm]
&=& \label{regroup}
\sum_{F_n\in B_n}\sum_{v\in V_n}
p^{|F_n|}(1-p)^{|E_n|-|F_n|}
\chi_{(v)}^{(F_n)}
\bigotimes_{e\in F_n\setminus F_n(v)}\phi^+_e
\bigotimes_{e\notin F_n}\none_e \nonumber\\[2mm]
&& +\sum_{F_n\in
C_n}p^{|F_n|}(1-p)^{|E_n|-|F_n|}\bigotimes_{e\in
F_n}\phi^+_e\bigotimes_{e\notin F_n}\none_e,\label{eq:summands2}
\end{eqnarray}
where
\begin{equation}
\chi_{(v)}^{(F_n)}=
\frac{1}{|G_n|}\bigotimes_{e\in F_n(v)}\phi^+_e+\frac{p^{-|F_n(v)|}(1-p)^{|F_n(v)|}}{2^{\delta_n(v)}-1}\bigotimes_{e\in F_n(v)}\none_e
\end{equation}
is an unnormalized state for each choice of $F_n\in B_n$ and $v\in V_n$ realizing the combination of $B_n$ and $A_n$ summands at each vertex. In fact, up to proper normalization, $\chi_{(v)}^{(F_n)}$ is an isotropic state in $H_1\otimes H_2$ where $\dim H_1=\dim H_2=d^{F_n(v)}$.

Now, according to Eq.~(\ref{isotropicthreshold}), such state is separable if and only if
\begin{equation}
\frac{\frac{1}{|G_n|}}{\frac{1}{|G_n|}+\frac{p^{-|F_n(v)|}(1-p)^{|F_n(v)|}}{2^{\delta_n(v)}-1}}\leq \frac{1}{d^{|F_n(v)|}+1}
\end{equation}
or, equivalently,
\begin{equation}
d^{|F_n(v)|}\Big(\frac{p}{1-p}\Big)^{|F_n(v)|}(2^{\delta_n(v)}-1)\leq |G_n|.
\end{equation}

We want to show that for every $p\in (0,1)$, there exists $n_0\geq 1$ such that for every $n\geq n_0$ the previous state is separable for all $F_n$ and $v$. Actually, according to Proposition~\ref{prop:facts}, we can assume that $p>1/2$, so that $p/(1-p)>1$.

Hence, it suffices to show that  for every $p\in (1/2,1)$, there exists $n_0\geq 1$ such that for every $n\geq n_0$ we have
\begin{equation}
d^{\delta_\mathrm{max}(G_n)}\Big(\frac{p}{1-p}\Big)^{\delta_\mathrm{max}(G_n)}2^{\delta_\mathrm{max}(G_n)}\leq |G_n|,
\end{equation}
or, equivalently,
\begin{equation}
\Big(\log d+ \log \frac{p}{1-p}+\log 2\Big)
  \delta_\mathrm{max}(G_n)\leq \log |G_n|,
\end{equation}
which follows from our assumption on the growth of $\delta_\mathrm{max}(G_n)$.
\end{proof}

The preceding theorem gives a sufficient condition for $\cG=(G_n)_n$ to be ABS based on the slow degree
growth of $\delta_{\mathrm{max}}(G_n)$. Note, nevertheless, that this result \textit{does not} imply the result of \cite{ASGME} that all trees
are ABS. To see this it is enough to consider the star graph $S_n$ (one
central vertex and $n$ leaves) since $S_n$ is a tree but $\delta_{\mathrm{max}}(S_n)=n$.
Using a similar proof as in the preceding theorem one can obtain the following
stronger version.

\begin{theorem}\label{th:degABSv2}
Let $\cG=(G_n)_{n\in\N}$ be a graph sequence for which there exists a choice of subsets of vertices $V'_n\subset V_n$ with $\lim_{n\to\infty}|V'_n|=\infty$ and such that
$\max_{v\in V'_n}\delta(v)\in o(\log|V_n'|)$. Then the graph sequence $\cG$ is ABS.
\end{theorem}
\begin{proof}
The proof follows the same lines as in the previous theorem with the sets $B_n$ and $C_n$ defined in the same way as therein (cf.\ Eqs.~(\ref{bn}) and (\ref{cn})) but with $A_n$ now containing only all subsets of edges labeling states that
are 1-biseparable \textit{exactly} at one vertex $v\in V'_n$, i.e.\
\begin{equation}\label{anprima}
A_n=\{F_n'\subseteq E_n \mid \textrm{ there exists a unique }
      v\in V'_n \textrm{ such that } F_n'(v)=\emptyset\}\;.
\end{equation}
Following the same reasoning that leads to Eq.~(\ref{eq:summands2}), we can see that $\sigma_d(G_n,p)$ can also be written as
\begin{eqnarray}
\sigma_d(G_n,p)&=& \label{regroup}
\sum_{F_n\in B_n}\sum_{v\in V'_n}
p^{|F_n|}(1-p)^{|E_n|-|F_n|}
\chi_{(v)}^{(F_n)}
\bigotimes_{e\in F_n\setminus F_n(v)}\phi^+_e
\bigotimes_{e\notin F_n}\none_e \nonumber\\[2mm]
&& +\sum_{F_n\in
C_n}p^{|F_n|}(1-p)^{|E_n|-|F_n|}\bigotimes_{e\in
F_n}\phi^+_e\bigotimes_{e\notin F_n}\none_e,
\end{eqnarray}
where now
\begin{equation}
\chi_{(v)}^{(F_n)}=
\frac{1}{|V'_n|}\bigotimes_{e\in F_n(v)}\phi^+_e+\frac{p^{-|F_n(v)|}(1-p)^{|F_n(v)|}}{2^{\delta_n(v)}-1}\bigotimes_{e\in F_n(v)}\none_e
\end{equation}
is an unnormalized isotropic state for each choice of $F_n\in B_n$ and $v\in V'_n$, and thus separable iff
\begin{equation}
d^{|F_n(v)|}\Big(\frac{p}{1-p}\Big)^{|F_n(v)|}(2^{\delta_n(v)}-1)\leq |V'_n|.
\end{equation}
Then, the same argument that concludes the proof of Theorem~\ref{th:degABS} leads us now to the claim here.
\end{proof}

\begin{remark}
Notice that Theorem~\ref{th:degABSv2} is strong enough to imply the result of \cite{ASGME} that all sequences of tree graphs are ABS. In fact, given a graph sequence  $\cG=(G_n)_{n\in\N}$ in which $G_n$ is a tree for all $n$, if we choose
\begin{equation}
 V_n'=\{v\in V_n\mid \delta(v)\leq 2 \}\;
\end{equation}
we obtain a sequence of subsets of vertices that satisfies the hypotheses of Theorem~\ref{th:degABSv2} as it must hold for all $n$ that
\begin{equation}
|V'_n|\geq\frac{|G_n|+2}{2}.
\end{equation}
This last claim follows from the fact that for every tree $T$ it holds that $|W|\leq(|T|-2)/2$, where $W$ is the subset of vertices of $T$ with degree strictly larger than 2. Indeed, for any tree $T=(V,E)$ one has that $|E|=|T|-1$ and, therefore, by the handshaking lemma that
\begin{equation}
2|T|-2=2|E|=\sum_{v\in V}\delta(v)\geq3|W|+|T|-|W|,
\end{equation}
which leads to the above inequality.
\end{remark}

\subsection{Fast degree growth implies the AGME property}\label{subsecdeg2}

We address in this subsection conditions on the minimal degree growth of a graph sequence $\cG=(G_n)_n$ that guarantee
that $\cG$ is AGME. To make contact with Section~\ref{seccon} recall that $\delta_{\mathrm{min}}(G_n)\geq \lambda(G_n)$.
Note, in addition, that in the context of simple graphs one always has $\delta_{\mathrm{min}}(G_n)< |G_n|$ so that the conditions
$\delta_{\mathrm{min}}(G_n)\in\Omega(|G_n|)$ or
$\lambda(G_n)\in\Omega(|G_n|)$ mean that the minimal degree or edge-connectivity
have essentially a linear growth, i.e.
$\delta_{\mathrm{min}}(G_n)\in\Theta(|G_n|)$ or $\lambda(G_n)\in\Theta(|G_n|)$ respectively.

We will state first two technical lemmas related with the growth of the edge-connectivity.

\begin{lemma}\label{lemma1degAGME}
Let $\cG=(G_n)_n$ be a graph sequence such that $\lambda(G_n)\not\in\Omega(|G_n|)$. Then, either $\lambda(G_n)\in o(|G_n|)$ or there exists a proper subsequence
$\cG_s=(G_{s(n)})_n$ such that $\lambda(G_{s(n)})\in\Omega(|G_{s(n)}|)$.
\end{lemma}
\begin{proof}
Note that $\lambda(G_n)\not\in\Omega(|G_n|)$ is equivalent to
\begin{equation}
\liminf_{n\to\infty}\frac{\lambda(G_n)}{|G_n|}=0.
\end{equation}
There are now two possibilities: either the limit of the sequence exists and
\begin{equation}
\lim_{n\to\infty}\frac{\lambda(G_n)}{|G_n|}=0\,,
\end{equation}
which means $\lambda(G_n)\in o(|G_n|)$, or there exists a subsequence
$\left(\frac{\lambda(G_{s(n)})}{|G_{s(n)}|}\right)_n$, satisfying
\begin{equation}
 \lim_{n\to\infty}\frac{\lambda(G_{s(n)})}{|G_{s(n)}|}\in (0,1]\;.
\end{equation}
Here, we have used that the sequence $(\lambda(G_n)/|G_n|)_n$ is bounded above by $1$.
Therefore we conclude that in this case
$\lambda(G_{s(n)})\in\Omega(|G_{s(n)}|)$.
\end{proof}

\begin{lemma}\label{lemma2degAGME}
 Let $\cG=(G_n)_{n\in\N}$ be a graph sequence such that
 $\lambda(G_n)\in o(|G_n|)$. If $\delta_{\mathrm{min}}\left({G}_n\right)\in\Omega(|G_n|)$ there exists a constant $c\in (0,1/2]$ such that
 \begin{equation}
 \delta_{\mathrm{min}}\left({G}_n\right)\geq c\left|{G}_n\right|
  +o\left(\left|{G}_n\right|\right)
  \end{equation}
and the two sequences of disjoint connected components
$\cG^{(1)}=\left(G_n^{(1)}\right)_n$ and $\cG^{(2)}=\left(G_n^{(2)}\right)_n$
obtained from $\cG$ by deleting for each $n$ a cut-edge set of cardinality
$\lambda(G_n)$ satisfy that $\delta_{\mathrm{min}}\left({G}_n^{(i)}\right)\in\Omega(|G_n^{(i)}|)$ and, in particular,
\begin{equation}
  \delta_{\mathrm{min}}\left({G}^{(i)}_n\right)\geq\frac{c}{1-c}\;\left|{G}^{(i)}_n\right|
  +o\left(\left|{G}^{(i)}_n\right|\right)\;,\quad i\in\{1,2\}.
\end{equation}
\end{lemma}

\begin{remark}
It is easy to see that, given two sequences $(x_n)_n$ and $(y_n)_n$ such that $\lim_n y_n=\infty$, then the condition $x_n\in\Omega (y_n)$ is equivalent to the existence of a universal constant $c>0$ and a sequence $(a_n)_n$ such that $a_n\in o(y_n)$ satisfying $x_n\geq cy_n+a_n$ for every $n$. In particular,
$x_n/y_n\geq c+ o(1)$. We have made this explicit in the previous lemma because the constants involved in the statement will be important later.
\end{remark}

\begin{proof}
The asymptotic assumption on $\delta_{\mathrm{min}}(G_n)$ is equivalent the fact that there is a positive constant $c$ such that
\begin{equation}\label{eq:cn}
c_n:=\frac{\delta_{\mathrm{min}}(G_n)}{|G_n|}\geq c+ o(1)\;.
\end{equation}
Moreover, note that we must have $0< c \leq 1/2$ since, otherwise, there exists a subsequence such that for all $n$ we have
$c_{s(n)} > 1/2$ and then, by Eq.~(\ref{maxedgeconnected}), $G_{s(n)}$
must be maximally edge-connected for all $n$ so that
$\lambda(G_{s(n)}) =\delta_{\mathrm{min}}(G_{s(n)})\in \Omega(|G_{s(n)}|)$ contradicting the hypothesis that $\lambda(G_n) \in o(|G_n|)$.

Consider now the graph sequences $\cG^{(i)}$, $i=1,2$, mentioned in the statement of the lemma and note that, by construction,
$\delta_{\mathrm{min}}\left(G_n^{(i)}\right)
\geq\delta_{\mathrm{min}}(G_n)-\lambda(G_n)$. We also consider the vertex relation
\begin{align}
 d_n^{(i)}:=\frac{|G_n^{(i)}|}{|G_n|}, \, \,  i=1,2\;,
 \end{align}
 which satisfies $d_n^{(1)}+d_n^{(2)}=1$. Moreover, since $\lambda(G_n)\in o(|G_n|)$, we have
\begin{equation}\label{eq:dn}
d_n^{(i)}\geq\frac{\delta_{\mathrm{min}}(G_n^{(i)})}{|G_n|}
         \geq  \frac{\delta_{\mathrm{min}}(G_n)}{|G_n|}-\frac{\lambda(G_n)}{|G_n|}\geq c+o(1)\,.
    \end{equation}

In addition, we can write for $d_n^{(1)}$ (and similarly for $d_n^{(2)}$)
 \begin{equation}\label{eq:1/dn}
 d_n^{(1)} =1- d_n^{(2)}\leq 1-c+o(1)\quadtext{and, hence,}
           \frac{1}{d_n^{(1)}}
           \geq \frac{1}{1-c}+o(1)\;.
 \end{equation}

Finally,  using (\ref{eq:dn}) and (\ref{eq:1/dn}) we have that, for $i=1,2$,
\begin{equation}
 \frac{\delta_{\mathrm{min}}(G_n^{(i)})}{|G_n^{(i)}|}=\frac{\delta_{\mathrm{min}}(G_n^{(i)})}{|G_n|}\frac{1}{d_n^{(i)}}\geq (c+o(1))\Big( \frac{1}{1-c}+o(1)\Big)=
 \frac{c}{1-c}+o(1),
\end{equation}
which concludes the proof.
\end{proof}

Note that, if $\lambda(G_n^{(i)})\in o\left(|G_n^{(i)}|\right)$ holds for some $i=1,2$, we could apply Lemma~\ref{lemma2degAGME} again on $\cG^{(i)}$ to obtain $\cG^{(i1)}=\left(G_n^{(i1)}\right)_n$ and $\cG^{(i2)}=\left(G_n^{(i2)}\right)_n$ such that, for $j=1,2$,
\begin{equation}
 \frac{\delta_{\mathrm{min}}(G_n^{(ij)})}{|G_n^{(ij)}|}\geq
 \frac{\frac{c}{1-c}}{1-\frac{c}{1-c}}+o(1)= \frac{c}{1-2c}+o(1).
\end{equation}
Iterating this process $k$ times assuming the previous asymptotic behavior of the edge connectivity holds each time and applying Lemma~\ref{lemma2degAGME} we obtain the constant
\begin{equation}\label{fc}
 f_c(k):=\frac{c}{1-kc}\;,
\end{equation}
which can be easily verified by induction. This motivates the following result on the behavior of the function $f_c(k)$ that will be used later. We omit the proof since it is straightforward (notice that $(0,1/2]=\cup_{j=2}^\infty \left(1/(j+1),1/j\right]$).

\begin{lemma}\label{lemma_f_c}
For $j=2,3,\dots$, if $c\in \left(\frac{1}{j+1}, \frac{1}{j}\right)$, then
 $f_c(k)\in \left(\frac{1}{j-k+1},\frac{1}{j-k}\right)$ for $1\leq k\leq j-1$ and
 if $c=\frac{1}{j}$, then $f_c(k)=\frac{1}{j-k}$ for every $k$.
\end{lemma}

The following proposition shows how a certain growth condition on the minimal degree implies structural properties for the graph sequences and the edge-connectivity growth.
In fact, we show that if $\delta_{\mathrm{min}}(G_n)\in\Omega(|G_n|)$, then the sequence splits into a uniformly bounded number of induced subgraphs (clusters) $G_n^{(l)}$ where the edge-connectivity in each cluster $l$ is in $\Omega(|G_n^{(l)}|)$ (see Fig.~\ref{fig:globos}).

\begin{figure}[h]
\centering
\begin{tikzpicture}[scale=0.8]
    \node[draw, circle, inner sep=1pt, minimum size=4pt, fill=black] (a1) at (0, 1) {};
    \node[draw, circle, inner sep=1pt, minimum size=4pt, fill=black] (a2) at (0, -1) {};

    \node[draw, circle, inner sep=1pt, minimum size=4pt, fill=black] (b1) at (2.75, 1) {};
    \node[draw, circle, inner sep=1pt, minimum size=4pt, fill=black] (b5) at (3.25, 1) {};
    \node[draw, circle, inner sep=1pt, minimum size=4pt, fill=black] (b3) at (2.75, -1) {};
    \node[draw, circle, inner sep=1pt, minimum size=4pt, fill=black] (b4) at (3.25, -1) {};
    \node[draw, circle, inner sep=1pt, minimum size=4pt, fill=black] (b2) at (2.75, 0) {};

    \node[draw, circle, inner sep=1pt, minimum size=4pt, fill=black] (c1) at (6, 1) {};
    \node[draw, circle, inner sep=1pt, minimum size=4pt, fill=black] (c2) at (6, -1) {};

    \draw[black] (0, 0) ellipse (1.25cm and 1.5cm);
    \draw[black] (3, 0) ellipse (1.25cm and 1.5cm);
    \draw[black] (6, 0) ellipse (1.25cm and 1.5cm);

    \draw (a1) -- (b1);
    \draw (a1) -- (b2);
    \draw (a2) -- (b3);

    \draw (b5) -- (c1);
    \draw (b4) -- (c2);

    \node at (0, -2.5) {$G_n^{(1)}$};
    \node at (3, -2.5) {$G_n^{(2)}$};
    \node at (6, -2.5) {$G_n^{(3)}$};

\end{tikzpicture}
 \caption{Element $G_n$ of the graph sequence $\cG$ split into three clusters
 $G_n^{(i)}$, $i=1,2,3.$}
\label{fig:globos}
\end{figure}

\begin{proposition}\label{propdegAGME}
Let $\cG=(G_n)_{n\in\N}$ be a graph sequence such that
$\delta_{\mathrm{min}}(G_n)\in\Omega(|G_n|)$.
Then, there exists a constant $K$ depending only on $\cG$, a subsequence $\cG_s=(G_{s(n)})_{n\in\N}$ and a partition of the vertex sets
\begin{equation}
V_{s(n)}=\bigsqcup_{l=1}^{K}V_{s(n)}^{(l)}
\end{equation}
such that if
$\cG_s^{(l)}=\left(G^{(l)}_{s(n)}\right)_{n\in\N}$,
$l\in\{1,\dots, K\}$, is the sequence of subgraphs induced by the partition of vertices
$\left(V_{s(n)}^{(l)}\right)_{n\in\N}$, then
\begin{equation}
\lambda\left(G_{s(n)}^{(l)}\right)\in\Omega\left(|G_{s(n)}^{(l)}|\right)\,.
\end{equation}
\end{proposition}
\begin{proof}
The procedure we follow here is based on two basic operations depending on the asymptotic behavior of the edge-connectivity $\lambda$ mentioned in Lemma~\ref{lemma1degAGME}:
passing to subsequences and passing to sequences of subgraphs.

By assumption, we know that there exists a constant $c > 0$ such that
$\delta_{\mathrm{min}}(G_n)\geq c|G_n|+ o(|G_n|)$. Moreover, we can assume that $c\in (0,1/2]$, since we can reason as in the proof of Lemma~\ref{lemma2degAGME} to conclude that, if $c>1/2$, we actually have
$\lambda(G_n)=\delta_{\mathrm{min}}(G_n)=\Omega(|G_n|)$. In this case we finish the proof by taking $K=1$ and $s(n)=n$ for every $n$.

We illustrate next the first nontrivial step of the aforementioned procedure.
If there is a subsequence of graphs $\cG_s = (G_{s(n)})_n$ such that
$\lambda(G_{s(n)}) \in\Omega(|G_{s(n)}|)$, then we again finish the proof with $K = 1$ (note that this includes the trivial subsequence $s(n) = n$ in case the original graph sequence already satisfies $\lambda(G_n) \in\Omega(|G_n|)$). On the other hand, according to
the second possibility in Lemma~\ref{lemma1degAGME}, if $\lambda(G_n)\in o(|G_n|)$, we then consider the two (disjoint) sequences of connected components
\begin{equation}
\cG^{(i)}=\left(G^{(i)}_n\right)_n \;, \;i=1,2\;,
\end{equation}
obtained from $\cG$ by deleting for each $n$ an edge-cut of cardinality
$\lambda(G_n)$ as in Lemma~\ref{lemma2degAGME}. From this lemma, we also obtain that
for $i=1,2$
\begin{equation}\label{prop651}
 \delta_{\mathrm{min}}(G_n^{(i)})\geq f_c(1) |G_n^{(i)}|+o\left(|G_n^{(i)}|\right)\,,
\end{equation}
where the function $f_c(k)$ is defined in Eq.~(\ref{fc}).

In order to proceed with the next step, we look at the edge-connectivity of the sequence of subgraphs $\cG^{(1)}$, for which, according, again, to Lemma~\ref{lemma1degAGME}, there are two possibilities:
\begin{itemize}
   \item [a)]
   If $\lambda(G_{n}^{(1)})\in o(|G_{n}^{(1)}|)$, then one splits again $\cG^{(1)}$
   by deleting an edge-cut of cardinality $\lambda(G_{n}^{(1)})$ for each $n$ obtaining two new sequences of connected components $\cG^{(11)}$ and
   $\cG^{(12)}$, with associated constants $f_c(2)$,  and moves on to analyze $\cG^{(11)}$.
 \item [b)] If there exists a subsequence $\cG_{s_1}^{(1)}$ such that
 \begin{equation}
 \lambda\left(G_{s_1(n)}^{(1)}\right)\in\Omega \left(|G_{s_1(n)}^{(1)}|\right),
 \end{equation}
 then we conclude the analysis of this branch and restrict the subsequence of the second component to $\cG_{s_1}^{(2)}$, which we continue analyzing. Two cases may now appear:
 \begin{itemize}
  \item [a)] If $\lambda\left(G_{s_1(n)}^{(2)}\right)\in o\left(|G_{s_1(n)}^{(2)}|\right)$, then the splitting procedure continues along this branch.
  \item [b)] If there exists a subsequence $s_2$ within $s_1$ such that
  $\lambda\left(G_{s_2(n)}^{(2)}\right)\in\Omega \left(|G_{s_2(n)}^{(2)}|\right)$, then the procedure stops completely and one has the subsequence of induced subgraphs
  $\cG_{s_2}^{(1)}\sqcup \cG_{s_2}^{(2)}$ with $K=2$ and
  $\lambda\left(G_{s_2(n)}^{(i)}\right)\in\Omega \left(|G_{s_2(n)}^{(i)}|\right)$.
  Note that one has to restrict in the first branch the subsequence $\cG_{s_1}^{(1)}$ to $\cG_{s_2}^{(1)}$. Since the growth condition passes to subsequences we have that if
  $\lambda\left(G_{s_1(n)}^{(1)}\right)\in\Omega \left(|G_{s_1(n)}^{(1)}|\right)$, then
  $\lambda\left(G_{s_2(n)}^{(1)}\right)\in\Omega \left(|G_{s_2(n)}^{(1)}|\right)$ as well.
 \end{itemize}
\end{itemize}

In full generality, each step in the process is labeled by a multi-index $\vec{a}=(a_1,\dots,a_k)$ with $a_l\in\{1,2\}$ ($l=1,\dots,k$), which the procedure follows in lexicographical order. The length of the multi-index, i.e.\ $|\vec{a}|:=k$, counts the number of splittings that have occurred to reach this point and the different values of $a_l$ describe to which branch of the procedure the step corresponds to. Notice that the length of the multi-index varies through this process.

At the beginning of an arbitrary step labeled by $\vec{a}$, we start with a subsequence of induced subgraphs $\cG_{s_{\vec{b}}}^{\vec{a}}=(G_{s_{\vec{b}}(n)}^{\vec{a}})_n$, where $\vec{b}$ is the multi-index that precedes $\vec{a}$ in lexicographical order (we just set $s_{\emptyset}=id$ for $\vec{b}=\emptyset$ the preceding multi-index for the first step $\vec{a}=1$). By Lemma~\ref{lemma2degAGME} we have that
$\delta_{\mathrm{min}}(G_{s_{\vec{b}}(n)}^{\vec{a}}) \in\Omega(|G_{s_{\vec{b}}(n)}^{\vec{a}}|)$ and, again, we have two possibilities:
\begin{itemize}
    \item[a)] The subgraph subsequence $\cG^{\vec{a}}_{s_{\vec{b}}}$ is such that
\begin{equation}
\lambda\left(G_{s_{\vec{b}}(n)}^{\vec{a}}\right)\in o\left(|G_{s_{\vec{b}}(n)}^{\vec{a}}|\right)\,.
\end{equation}

\item[b)] There is a subsequence $s_{\vec{a}}$ within $s_{\vec{b}}$ such that the corresponding subgraph subsequence $\cG^{\vec{a}}_{s_{\vec{a}}}$ fulfills
\begin{equation}
\lambda\left(G_{s_{\vec{a}}(n)}^{\vec{a}}\right)
\in\Omega\left(|G_{s_{\vec{a}}(n)}^{\vec{a}}|\right).
\end{equation}
\end{itemize}
In case b) the procedure stops in this branch, all previous subgraph subsequences are
restricted to this one, i.e.\ we fix $s_{\vec{a}'}=s_{\vec{a}}$ for every multi-index
$\vec{a}'$ that precedes $\vec{a}$ in lexicographical order and we move on to consider
$\cG^{\vec{a}''}_{s_{\vec{a}}}$, where $\vec{a}''$ is the subsequent multi-index to
$\vec{a}$ in lexicographical order. In case a) we define $s_{\vec{a}}=s_{\vec{b}}$ and we create two new multi-indices $(\vec{a},1)$ and $(\vec{a},2)$ (where $(\vec{a},i):=(a_1,\dots,a_k,i)$ for $i=1,2$), which correspond to splitting $\cG_{s_{\vec{a}}}^{\vec{a}}$ into two connected subgraph components $\cG_{s_{\vec{a}}}^{(\vec{a},1)}$ and $\cG_{s_{\vec{a}}}^{(\vec{a},2)}$ deleting the corresponding cut-edge set. We move on then to analyze the subgraph subsequence corresponding to the subsequent multi-index to $\vec{a}$ in lexicographical order, i.e.\ $(\vec{a},1)$.

We see that the analysis of a given branch of the procedure concludes when the process reaches condition b) above.  In this case, after readjusting the previous subsequences if necessary, we conclude the analysis of this branch of the process. Hence, in order to finish the proof, we must show that this process necessarily ends in a finite number of steps. We consider two cases:
 \begin{itemize}
 \item
 If $c\in \left(\frac{1}{j+1}, \frac{1}{j}\right)$ for some $j\geq 2$, according to Lemma \ref{lemma_f_c}, we will have that $f_c(k)\in (0,\frac{1}{2})$ for every $1\leq k\leq j-2$ and $f_c(j-1)\in (1/2,1)$. Hence, we know that for every branch described by $\vec{a}$ such that $|\vec{a}|=j-1$, the corresponding subgraphs $(G_{s_{\vec{b}}(n)}^{\vec{a}})_n$ satisfy that their edge-connectivity equals their minimal degree. So that we are in the case b) above.
 \item  If $c=\frac{1}{j}$ for some $j\geq 2$, we know that $f_c(k)\in (0,1/2)$ for every $1\leq k\leq j-2$ and  $f_c(j-1)=1/2$.
 Hence, we know that for every branch described by $\vec{a}$ such that $|\vec{a}|=j-1$, the corresponding subgraphs $(G_{s_{\vec{b}}(n)}^{\vec{a}})_n$ satisfy

\begin{equation}\label{prop652}
 \delta_{\mathrm{min}}(G_{s_{\vec{b}}(n)}^{\vec{a}})
 \geq \frac{1}{2}|G_{s_{\vec{b}}(n)}^{\vec{a}}|+ a_{s_{\vec{b}}(n)}\;,
\end{equation}
 where $a_{s_{\vec{b}}(n)}\in o(|G_{s_{\vec{b}}(n)}^{\vec{a}}|)$. Now, if
 Eq.~(\ref{prop652}) can only hold with sequences $(a_{s_{\vec{b}}(n)})_n$ such that $a_{s_{\vec{b}}(n)}<0$ for $n$ large enough, then it might still happen that

 $\lambda\left(G_{s_{\vec{b}}(n)}^{\vec{a}}\right)\in o\left(|G_{s_{\vec{b}}(n)}^{\vec{a}}|\right)$ and another splitting is necessary. After this, we would have
 \begin{equation}
\delta_{\mathrm{min}}(G_{s_{\vec{a}}(n)}^{(\vec{a},i)})\geq |G_{s_{\vec{a}}(n)}^{(\vec{a},i)}|+ b^{(i)}_{s_{\vec{a}}(n)},\hspace{0.2 cm} i=1,2,
\end{equation}
where $b^{(i)}_{s_{\vec{a}}(n)}\in o(|G_{s_{\vec{a}}(n)}^{(\vec{a},i)}|)$ for $i=1,2$.

 Therefore, we know that for every branch described by $\vec{a}$ such that $|\vec{a}|=j$, the corresponding subgraphs $(G_{s_{\vec{b}}(n)}^{\vec{a}})_n$ satisfy that their edge-connectivity equals their minimal degree and we are again in the case b) above.

 If, on the other hand, the above assumption on the sequence $(a_{s_{\vec{b}}(n)})_n$ in Eq.\ (\ref{prop652}) does not hold, then there must exist a subsequence $s_{\vec{a}}$ within $s_{\vec{b}}$ for which $a_{s_{\vec{a}}(n)}\geq0$ for every $n$ and, hence,
\begin{equation}
 \delta_{\mathrm{min}}(G_{s_{\vec{a}}(n)}^{\vec{a}})
         \geq\frac{1}{2}|G_{s_{\vec{a}}(n)}^{\vec{a}}|
\end{equation}
and we are, again, in case b).
\end{itemize}
Finally, the maximum number of splittings in the procedure described
in this proof is $j$ and ,therefore, we have that  $K\leq 2^j$, which only depends on $c$.
\end{proof}

\begin{proposition}\label{pro:degAGME}
If the graph sequence $\cG=(G_n)_n$ fulfills that
$\delta_{\mathrm{min}}(G_n)\in\Omega(|G_n|)$,
then $\cG$ is not ABS.
\end{proposition}
\begin{proof}
Let us assume, for a contradiction, that $\cG$ is ABS. Then, so must be every subsequence and, in particular, the subsequence $\cG_s$ obtained from Proposition~\ref{propdegAGME}, which must exist due to our hypothesis on $\delta_{\mathrm{min}}(G_n)$. Hence, for any fixed $d$ and for any choice of $p<1$, our assumption is that $\sigma_d(G_{s(n)},p)$ is biseparable for $n$ large enough, i.e., in this case we have from Eq.~(\ref{def:biseparable})
\begin{equation}\label{eqth11}
\sigma_d(G_{s(n)},p)=\sum_{\emptyset\neq M\subsetneq V_{s(n)}}q_M^{(n)}\chi_M^{(n)},
\end{equation}
where $\chi_M^{(n)}$ is an $|G_{s(n)}|$-partite quantum state that is separable in the bipartition $M|\overline{M}$ and for each $M,n$ we have $q_M^{(n)}\geq0$ with normalization $\sum_Mq_M^{(n)}=1$ (note that to simplify notation we omitted in some cases the explicit dependence on $d$ and $p$).

In the following it will be convenient to distinguish two different kinds of bipartitions. If $K$ is the constant in the statement of Proposition~\ref{propdegAGME},  for each $k\in [K]$ we say that a nontrivial (i.e., a nonempty and
proper) $M\in I_k^n$ if $M\subset V_{s(n)}$ and there exist $u,v\in V_{s(n)}^{(k)}$ such that $u\in M$ and $v\notin M$; define then
\begin{equation}\label{eq:Ik}
I^n=\bigcup_{k\in[K]}I^n_k.
\end{equation}
Thus, Eq.~(\ref{eqth11}) can be rewritten as
\begin{equation}\label{eqth11bis}
\sigma_d(G_{s(n)},p)=\sum_{M\in I^n}q_M^{(n)}\chi_M^{(n)}+\sum_{M\notin I^n}q_M^{(n)}\chi_M^{(n)}.
\end{equation}
Note that $M\notin I^n$ means that $M$ either contains all or none of the vertices
$V_{s(n)}^{(k)}$ of any $k$-cluster with $k\in [K]$. In other words, $M\notin I^n$ if and only if there exists $J\subsetneq [K]$ such that $M=\cup_{j\in J}V_{s(n)}^j$.

Now, if for a given $k$ we write
\begin{equation}\label{eqth11bis_k}
\sigma_d(G_{s(n)},p)=\sum_{M\in I_k^n}q_M^{(n)}\chi_M^{(n)}+\sum_{M\notin I_k^n}q_M^{(n)}\chi_M^{(n)}
\end{equation}and we take the partial trace with respect to
$\bigotimes_{e\notin E_{s(n)}(V_{s(n)}^{(k)})}H_e$ on both sides of Eq.~(\ref{eqth11bis_k}), we obtain
\begin{equation}
\sigma_d(G_{s(n)}^{(k)},p)=\sum_{M\in I_k^n}q_M^{(n)}\rho_{k,M}^{(n)}+\sum_{M\notin I_k^n}q_M^{(n)}\tau_{k,M}^{(n)},
\end{equation}
where the states $\{\rho_{k,M}^{(n)}\}_n$ are all separable in the bipartition\footnote{Here and in the following, we will simplify the notation. Note that we should say that the states $\{\rho_{k,M}^{(n)}\}_n$ are all separable in the bipartition $M_k|\overline{M}_k$, where $M_k=M\cap V_{s(n)}^k$ and $\overline{M}_k=\overline{M}\cap V_{s(n)}^k$.} $M|\overline{M}$. Then, as shown in the proof of Proposition~\ref{th:lambda}, graphs with edge-connectivity with linear growth satisfy the premises of Proposition~\ref{prop:sufficientASGME} and, therefore, by
Lemma~\ref{lemmaBSAto0} there exists $p_0<1$ such that if $p\geq p_0$ it must then hold that
\begin{equation}
\sum_{M\in I_k^n}q_M^{(n)}\xrightarrow[n\to\infty]{} 0\;.
\end{equation}
Considering this reasoning for every $k$ we then obtain that if $p$ is sufficiently close to 1, we have
\begin{equation}\label{argu1}
\sum_{M\in I^n}q_M^{(n)}\leq\sum_{k=1}^{K}\sum_{M\in I^n_k}q_M^{(n)}\leq K\max_k\sum_{M\in I^n_k}q_M^{(n)} \xrightarrow[n\to\infty]{} 0\;.
\end{equation}

We focus now on the second summand in Eq.~(\ref{eqth11bis}). Notice that, our previous characterization of the elements $M\notin I^n$ implies that the complement of $I^n$ has cardinality $2^{K-1}-1$ (excluding the empty set and taking into account the redundancy of the bipartitions
$M|\overline{M}$ and $\overline{M}|M$).
Now, for each $S\notin I^n$ we take the partial trace with respect to $\bigotimes_{e\notin E_{s(n)}(S,\overline{S})}H_e$ in both sides of Eq.~(\ref{eqth11bis}) to obtain that
\begin{equation}
\bigotimes_{e\in E_{s(n)}(S,\overline{S})}\rho_{e}(p,d)=q_S^{(n)}\eta_S^{(n)}+\sum_{M\neq S}q_M^{(n)}\omega_{M,S}^{(n)},
\end{equation}
where $\eta_S^{(n)}$ is separable in the bipartition $S|\overline{S}$. Then, viewing the state above as a bipartite state in $S|\overline{S}$
we have the following estimates in terms of the best separable approximation
(cf.\ Eq.~(\ref{bsa2}))
\begin{align}
q_S^{(n)}&\leq \mathrm{BSA}\left(\bigotimes_{e\in E_{s(n)}(S,\overline{S})}\rho_{e}(p,d)\right)
=\mathrm{BSA}\left(\rho(p,d)^{\otimes |E_{s(n)}(S,\overline{S})|}\right)\nonumber\\&\leq \mathrm{BSA}(\rho(p,d))
=\frac{d+1}{d}(1-p),
\end{align}
where the first inequality follows from the definition of the best separable approximation (cf.~Eq.~(\ref{bsa2})), the second from the fact that $|E_{s(n)}(S,\overline{S})|\geq 1$ for all $S\notin I^n$ and the last equality is Eq.~(\ref{bsaisotropic}), which holds as long as $p\geq1/(d+1)$. Thus,
\begin{equation}
\sum_{M\notin I^n}q_M^{(n)}\leq\left(2^{K-1}-1\right)\frac{d+1}{d}(1-p).
\end{equation}
Hence, we can choose a fixed value of $p$ sufficiently close to 1 so that for $n$ large enough
\begin{equation}\label{argu2}
\sum_{M\notin I^n}q_M^{(n)}\leq \frac{1}{2}\,.
\end{equation}
Consequently, Eqs.~(\ref{argu1}) and (\ref{argu2}) imply that the normalization
condition $\sum_Mq_M^{(n)}=1$ can no longer hold for $n$ large enough. We have therefore reached a contradiction and Eq.~(\ref{eqth11}) cannot hold for any fixed $d$ for all $p<1$ for sufficiently large $n$.
\end{proof}

\begin{theorem}\label{th:degAGME}
Any graph sequence $\cG=(G_n)_{n\in\N}$ such that $\delta_{\mathrm{min}}(G_n)\in\Omega(|G_n|)$ is AGME.
\end{theorem}
\begin{proof}
Assume for a contradiction that $\cG$ is not AGME. Then, according to Proposition~\ref{prop:notAGMEnotABS}~(i), there  exists a subsequence $\cG_s$ that is ABS. But, this is not possible by Proposition~\ref{pro:degAGME}, since  this subsequence still fulfills the same growth condition $\delta_{\mathrm{min}}(G_{s(n)})=\Omega(|G_{s(n)}|)$.
\end{proof}

\begin{remark}
Notice that, given that $\delta_{\mathrm{min}}(G)\geq\lambda(G)$ for every graph $G$, Theorem~\ref{th:degAGME} is stronger than Proposition~\ref{th:lambda} (i.e.\ the statement in the proposition follows from the statement in the theorem). However, the proof of Theorem~\ref{th:degAGME} uses Proposition~\ref{th:lambda}, so the latter cannot be skipped. The same happens with Proposition~\ref{th:lowlambda}, but, in this case, our proof of Theorem~\ref{th:degAGME} is independent of the arguments therein. Thus, the proof of this proposition in Section~\ref{seccon} could have been postponed to this section as the graph sequence considered therein is a particular instance of those that satisfy the premise of Theorem~\ref{th:degAGME}. However, for expository reasons we decided to consider this particular case in Section~\ref{seccon} in order to make clear at that stage the limitations of the edge-connectivity in our problem.
\end{remark}

\section{Explicit constructions}\label{subsecdeg3}

In this section we present several explicit constructions of graph sequences that are either AGME or ABS in order to analyze the tightness of our main results in the previous section and to show the limitations of the degree growth to characterize these properties. In addition to this, we briefly discuss the possibility of using other graph parameters to address this problem, focusing on the particular case of the diameter.

\subsection{Sharpness of the main results}

The following theorem shows that the sufficient condition for ABS based on the degree growth in Theorem~\ref{th:degABS} is not necessary and that the analogous sufficient condition for AGME in Theorem~\ref{th:degAGME} cannot be improved.

\begin{figure}[h]
\centering
\begin{tikzpicture}[scale=0.9]
    \foreach \i/\j in {0/1, 1/2, 2/3} {
        \node[draw, circle, inner sep=1pt, minimum size=4pt, fill=black] (a\i) at (\i*4, 0) {};  
        \draw[black] (\i*4, 0) ellipse (1cm and 1.25cm);  

        \node at (\i*4, 0.5) {$K_n$};

        \node at (\i*4, -1.755) {$G_n^{(\j)}$};
    }

    \draw (a0) -- (a1);
    \draw (a1) -- (a2);

\end{tikzpicture}
 \caption{Element $G_n$ of the graph sequence $\cG$ split into three clusters whose induced subgraphs are complete graphs $K_n$ with a bridge edge joining them.}
\label{fig:knbridge}
\end{figure}

\begin{theorem}\label{th:highdegABS}
For any function $f\colon\mathbb{N}\to\mathbb{N}$ such that $f(n)\in o(n)$, there exists an ABS graph sequence $\mathcal{G}=(G_n)_{n\in\mathbb{N}}$ such that
$\delta_{\mathrm{min}}(G_n)\in\Omega(f(|G_n|))$.
\end{theorem}
\begin{proof}
Define for every $n\in\mathbb{N}$ the positive function $g(n)=n/f(n)$ and notice first that if $f(n)\in o(n)$, then
\begin{equation}\label{eq:highdegABS}
    \lim_{n\to\infty} g(n)=\lim_{n\to\infty}\frac{n}{f(n)}=\infty,
\end{equation}
that is, $g(n)\in\omega(1)$.

We will now construct explicitly the desired sequence of graphs $\cG=(G_n)_n$ given $g$ (and, thus, given $f$).
For this (see Fig.~\ref{fig:knbridge}), we consider a function $k\colon\mathbb{N}\to\mathbb{N}$ that will be determined later; for each $n\in\N$ the graph $G_n$ is of order
$|G_n|=nk(n)$ and we partition the vertex set into $k(n)$ clusters of
order $n$
\begin{equation}
V_n=\bigsqcup_{m=1}^{k(n)}V_n^{(m)}\,,
\end{equation}
where $|V_n^{(m)}|=n$, for $m\in [k(n)]$. The edge set is specified as follows: for each $m$ the vertices in $V_n^{(m)}$ are all pairwise connected, i.e.
the induced subgraphs of each cluster correspond to complete graphs $K_n$.
Moreover, for each $m\in \{ 1,2,\dots,k(n)-1\}$ we can numerate the vertices
so that there is a unique pair of vertices $v_n^{(m)}\in V_n^{(m)}$,
$v_n^{(m+1)}\in V_n^{(m+1)}$ with $\{v_n^{(m)},v_n^{(m+1)}\}\in E_n$,
i.e., neighbouring clusters have a unique bridge edge. By construction
we have that $\delta_{\mathrm{min}}(G_n)=n-1$. Moreover, we have
\begin{equation}
f(|G_n|)=f(nk(n))=\frac{nk(n)}{g(nk(n))}
\end{equation}
and we introduce the nondecreasing function
$\tilde{g}$ given by
\begin{equation}
\tilde{g}(n)=\inf_{k\geq n}g(k)
\end{equation}
which, by construction, satisfies $g(n)\geq \tilde{g}(n)$, $n\in\mathbb{N}$, and
has the same asymptotic behavior as $g$, i.e., $\tilde{g}(n)\in\omega(1)$.
Therefore, for any $n\in\N$ we have
\begin{equation}
\frac{g(nk(n))}{k(n)}\geq\frac{\tilde{g}(nk(n))}{k(n)}
                     \geq\frac{\tilde{g}(n)}{k(n)}\;,
\end{equation}
and, hence,
\begin{equation}
    \frac{\delta_{\mathrm{min}}(G_n)}{f(|G_n|)}\geq\frac{(n-1)\tilde{g}(n)}{nk(n)}\;.
\end{equation}
Thus, choosing $k(n)=\lceil\tilde{g}(n)\rceil\in\omega(1)$ we obtain that for $n$
large enough it holds that \linebreak $\delta_{\mathrm{min}}(G_n)/f(|G_n|)\geq 1/4$, and, consequently, $\delta_{\mathrm{min}}(G_n)\in\Omega(f(|G_n|))$.

To conclude the proof we still need to show that $\cG$ is ABS. To this end, notice that our original Hilbert space $H=\bigotimes_{i=1}^{|G_n|} H_i$ can be rewritten as $H\cong H'=\bigotimes_{m=1}^{k(n)} H'_m$ where $H'_m=\bigotimes_{i\in V_n^{(m)}} H_i$ (i.e.\ we consider all parties in the cluster $V_n^{(m)}$, $m\in [k(n)]$, as a single party). An immediate but key observation here is that if our state is in $\mathrm{BS}(H')$ then it is in $\mathrm{BS}(H)$. Interestingly, our isotropic network state $\sigma_d(G_n,p)\in D(H)$ when viewed in this way it is written as
\begin{equation}
\sigma_d(L_{k(n)},p)\otimes\bigotimes_{m=1}^{k(n)}\tau_m\in D(H'),
\end{equation}
where each state $\tau_m$ acts on a subspace of $H'_m$ (these states come from the isotropic states corresponding to the edges within each cluster) and $L_n$ is the path graph with $n$ vertices (and, in particular, a tree for every $n$).
By our choice of $k$ above, which has the property that $k(n)\in\omega(1)$, we have then that $|L_{k(n)}|\to\infty$ as $n\to\infty$. Thus, by the results in \cite{ASGME} (or by using Theorem~\ref{th:degABS} here),
for any given $d$ we have the property that for any $p<1$ we have
$\sigma_d(L_{k(n)},p)$ is biseparable for $n$ large enough and the proof is concluded.
\end{proof}

Next we show that the sufficient condition for AGME of
Theorem~\ref{th:degAGME} is not necessary and, in combination with
Theorem~\ref{th:highdegABS}, that the degree growth cannot completely characterize AGME or ABS. Additionally, while it does not imply that the sufficient condition for ABS given in Theorem~\ref{th:degABS} cannot be improved, it reduces drastically the gap left for improvement. 

\begin{theorem}\label{th:lowdegAGME}
For any choice of $\alpha\in(0,1]$, there exists an AGME sequence of regular graphs $\cG=(G_n)_{n\in\N}$
such that $\delta_{\max}(G_n)\in O(|G_n|^\alpha)$.
\end{theorem}
\begin{proof}
The graph sequence consists of graphs
$G_n=(V_n,E_n)$ of order $|G_n|=n^k$ and it will be convenient to use multi-indices to numerate the elements of $V_n$: for any choice of $k\in\mathbb{N}$ consider
\begin{equation}
V_n=\{I=(i_1,i_2,\ldots,i_k)\mid i_m\in[n]\quadtext{and} m\in[k]\}\;.
\end{equation}
The edge set $E_n$ is given as follows: $\{I,J\}\in E_n$ if $i_m=j_m$ for
exactly $k-1$ different values of $m$ (i.e.\ when the Hamming distance of two
multi-indices is one) \footnote{Formally, our choice of $G_1$ is ill-defined because $|G_1|=1$. This is obviously not a problem and we can fix by hand $G_1$ to be any graph of choice, for instance $G_1=K_2$ for every choice of $k$.}. By construction, for any $v\in V_n$ it holds that
\begin{equation}
\delta(v)=k(n-1)=k\left(|G_n|^{1/k}-1\right).
\end{equation}

Hence, for any choice of $\alpha\in(0,1]$ we can take $k$ large enough so that the corresponding sequence of regular graphs $\cG$ as defined above has degree growth which is in $O(|G_n|^\alpha)$. Therefore, it remains to prove that for any choice of $k\in\mathbb{N}$ the graph sequence $\cG$ defined in this way is AGME. For this, it is enough to check that the premises of Proposition~\ref{prop:sufficientASGME} are met. This is done in the following by showing that for every pair of vertices of $G_n$ there exist at least $n-1$ edge-disjoint paths with length at most $k+1$ (independently of $n$) connecting them. Notice that $n-1=|G_n|^{1/k}-1$
satisfies the needed condition since
\begin{equation}
 \frac{n-1}{\log(|G_n|)}=\frac{n-1}{k\log(n)}\xrightarrow[n\to\infty]{} \infty\,.
\end{equation}
We finish then the proof by explicitly describing the paths needed.
For the sake of simplicity, consider first two arbitrary vertices $I$ and $J$ such that $i_m\neq j_m$, for every $m\in [k]$.
Then, there is an obvious path of length $k$ connecting them, namely
\begin{equation}\label{path1}
I\rightarrow(j_1,i_2,\ldots,i_k)\rightarrow(j_1,j_2,i_3,\ldots,i_k)\rightarrow\cdots\rightarrow J.
\end{equation}
Now, in addition, for every $i'_1\neq i_1,j_1$ we can consider the following $n-2$ paths of length $k+1$
\begin{equation}\label{path2}
I\rightarrow(i'_1,i_2,\ldots,i_k)\rightarrow(i'_1,j_2,i_3,\ldots,i_k)\rightarrow\cdots\rightarrow(i'_1,j_2,j_3,\ldots,j_k)\rightarrow J.
\end{equation}
The fact $i'_1$ takes $n-2$ different values guarantees that all paths in Eq.~(\ref{path2}) are edge-disjoint among themselves while the fact that $i'_1\neq j_1$ guarantees that they remain edge-disjoint together with the path in Eq.~(\ref{path1}). Finally, notice that if the multiindices $I$ and $J$ of the vertices agree on $l<k$ entries the same construction works leading to the same number of paths of smaller lengths. In fact, the path corresponding to Eq.~(\ref{path1}) has now length $k-l$. For the remaining $n-2$ paths corresponding to Eq.~(\ref{path2}), let $p$ be the first entry for which $I$ and $J$ disagree and consider the analogous construction for every $i'_p\neq i_p,j_p$, resulting in edge-disjoint paths of length $k-l+1$.
\end{proof}

\subsection{Diameter growth}\label{sec:diameter}

Given the limitations of the connectivity and the degree to fully characterize which graph sequences are AGME or ABS, it is natural to ask whether there are other graph parameters that can better apprehend these behaviors. Although this falls out of the scope of the present work, as an invitation for future research, we consider here the particular case of the diameter, which is clearly related to connectivity properties of the network as it quantifies how far parties can be in graph distance. In this regard, it is interesting to observe that all graph sequences that can be catalogued as AGME with the results presented so far must have a uniformly bounded diameter. This follows from Eq.~(\ref{ineqdiam}) for those that satisfy the condition of Theorem~\ref{th:degAGME} and from the very definition of diameter for those that fulfill that of Proposition~\ref{prop:sufficientASGME} (such as the AGME graph sequences of Theorem~\ref{th:lowdegAGME}). Thus, this could lead us to conjecture that a uniformly bounded diameter is a necessary condition for AGME (it certainly cannot be sufficient as the example of a sequence of star graphs of increasing order shows). However, in the following we show that this is not the case by providing a construction of an AGME graph sequence with unbounded diameter. For this, we need first the following lemma, which is an immediate consequence of the main result of \cite{beigi}.

\begin{lemma}\label{lem:tensor-bsa}
Let $\rho\in D(H_1\otimes H_2)$ be entangled, then
\begin{equation}
\lim_{n\to\infty}\mathrm{BSA}\left(\rho^{\otimes n}\right)=0.
\end{equation}
\end{lemma}
\begin{proof}
Let
\begin{equation}
T(\rho)=\frac{1}{2}\min_{\chi\in S(H_1\otimes H_2)}\|\rho-\chi\|,
\end{equation}
where $\|\cdot\|$ stands for the trace norm. This statement follows immediately from the result of \cite{beigi} that for every entangled state it holds that
\begin{equation}
\lim_{n\to\infty}T\left(\rho^{\otimes n}\right)=1,
\end{equation}
using that for any state $\rho$
\begin{equation}\label{bsadistance}
\mathrm{BSA}(\rho)\leq1-T(\rho).
\end{equation}
To see this last claim, let $\sigma$ be a state and $\tau$ a separable state so that
\begin{equation}
\rho=(1-\mathrm{BSA}(\rho))\sigma+\mathrm{BSA}(\rho)\tau.
\end{equation}
Then, for any separable state $\tau'$ we have the estimate
\begin{equation}
T(\rho)\leq\frac{1}{2}\|\rho-\mathrm{BSA}(\rho)\tau-(1-\mathrm{BSA}(\rho))\tau'\|
      =\frac{(1-\mathrm{BSA}(\rho))}{2}\|\sigma-\tau'\|\leq 1-\mathrm{BSA}(\rho)
\end{equation}
which concludes the proof.
\end{proof}

\begin{theorem}\label{teo:diam}
There exists an AGME graph sequence $\cG=(G_n)_{n\in\N}$ with the property that $\mathrm{diam}(G_n)\notin O(1)$.
\end{theorem}
\begin{proof}
The construction of the sequence of graphs (see Fig.~\ref{fig:k3}) depends on a function $f\colon\mathbb{N}\to\mathbb{N}$ such that $\lim_{n\to\infty}f(n)=\infty$ and that will be specified  later to have a sufficiently slow growth. The graph $G_n$ has order
$|G_n|=f(n)n$ and we partition again the vertex set into $f(n)$ clusters of order $n$ by
\begin{equation}
V_n=\bigsqcup_{k=1}^{f(n)}V_n^{(k)},
\end{equation}
where $|V_n^{(k)}|=n$ for all $k\in [f(n)]$. We label the vertex set as
ordered pairs $( i,k)\in [n]\times [f(n)]$ and the edges are given as follows:
\begin{itemize}
 \item[(i)] for each $k\in [f(n)]$ we have
 $\left\{( i,k), (i',k)\right\}\in E_n$ if $i<i'$, i.e.,
 the subgraphs induced by the cluster vertices correspond to complete graphs $K_n$;
 \item[(ii)] for each $k\in \{1,\dots, f(n)-1\}$, $i\in [n]$ we have
 $\left\{( i,k), ( i,k+1)\right\}\in E_n$, i.e., these edges specify an edge-cut of cardinality $n$ and each vertex in $V_n^{(k)}$,
 $k=2,\dots,f(n)-1$ shares two such edges with the ``twin'' vertices in the neighboring clusters.
\end{itemize}

\begin{figure}[h]
\centering
\begin{tikzpicture}[scale=0.9]
    \foreach \i in {0,1,2} {
        \node[draw, circle, inner sep=1pt, minimum size=4pt, fill=black] (a\i) at (\i*3 - 0.25, 0) {};  
        \node[draw, circle, inner sep=1pt, minimum size=4pt, fill=black] (b\i) at (\i*3 + 0.75, 1) {};  
        \node[draw, circle, inner sep=1pt, minimum size=4pt, fill=black] (c\i) at (\i*3 + 0.75, -1) {};  

        \draw (a\i) -- (b\i);
        \draw (a\i) -- (c\i);
        \draw (b\i) -- (c\i);

        \draw[black] (\i*3 + 0.5, 0) ellipse (1.2cm and 1.75cm);
    }

    \draw[bend left=0] (a0) to (a1);
    \draw[bend left=0] (a1) to (a2);
    \draw[bend left=0] (b0) to (b1);
    \draw[bend left=0] (b1) to (b2);
    \draw[bend left=0] (c0) to (c1);
    \draw[bend left=0] (c1) to (c2);
\end{tikzpicture}
 \caption{Element $G_3$ of the graph sequence $\cG$ split into three clusters whose induced subgraphs are $K_3$.}
\label{fig:k3}
\end{figure}

Notice then that, by construction, $\mathrm{diam}(G_n)=f(n)-1$ and, then, by the fact that \linebreak $\lim_{n\to\infty}f(n)=\infty$, we have $\textrm{diam}(G_n)\notin O(1)$.
It remains to show in the following that $\cG$ is AGME for some choice of $f$.
From here on, the proof follows the reasoning of
Proposition~\ref{pro:degAGME}. Assume for a contradiction that $\cG$ is not AGME and, therefore, that it contains an ABS subsequence $\cG_s=(G_{s(n)})_n$. Hence, like in the proof of Proposition~\ref{pro:degAGME}, for any fixed $d$ and for any choice of $p<1$ we have that $\sigma_d(G_{s(n)},p)$ is biseparable for $n$ large enough, i.e., in this case
\begin{equation}\label{eqapp}
\sigma_d(G_{s(n)},p)=\sum_{\emptyset\neq M\subsetneq V_{s(n)}}q_M^{(n)}\chi_M^{(n)},
\end{equation}
where $\chi_M^{(n)}$ is an $|G_{s(n)}|$-partite quantum state that is separable in the bipartition $M|\overline{M}$ and for any $M,n$ we have $q_M^{(n)}\geq 0$ with
$\sum_M q_M^{(n)}=1$. Moreover,
\begin{equation}\label{eqappbis}
\sigma_d(G_{s(n)},p)=\sum_{M\in I^n}q_M^{(n)}\chi_M^{(n)}+\sum_{M\notin I^n}q_M^{(n)}\chi_M^{(n)},
\end{equation}
where we divide as in Proposition~\ref{pro:degAGME} the partitions into two types
$I^n= \cup_{k=1}^{f(n)}I^n_k$ and its complement (see also Eq.~(\ref{eq:Ik})).
Now, if for a given $k$ we take the partial trace with respect to the subspace
$\bigotimes_{e\notin E_{s(n)}\left(V_{s(n)}^{(k)}\right)}H_e$ on both sides of Eq.~(\ref{eqappbis}), we obtain
\begin{equation}
\sigma_d(K_{s(n)},p)=\sum_{M\in I_k^n}q_M^{(n)}\rho_{k,M}^{(n)}+\sum_{M\notin I_k^n}q_M^{(n)}\tau_{k,M}^{(n)},
\end{equation}
where the states $\{\rho_{k,M}^{(n)}\}_n$ are separable in the bipartition
$M|\overline{M}$. Thus, as in the proof of Proposition~\ref{pro:degAGME}, there exists $p_0<1$ such that if $p\geq p_0$ it must then hold that
\begin{equation}\label{argu1app}
\sum_{M\in I^n}q_M^{(n)}\leq\sum_{k=1}^{f(n)}\sum_{M\in I_k^n}q_M^{(n)}\leq f(n)\max_k\sum_{M\in I_k^n}q_M^{(n)}\xrightarrow[n\to\infty]{} 0\;,
\end{equation}
if we choose $f$ to grow slow enough.

We focus now on the second summand in Eq.~(\ref{eqappbis}) and notice that the complement of $I^n$ can be taken to have cardinality $2^{f(n)-1}-1$; for each $S\notin I^n$ we take the partial trace with respect to $\bigotimes_{e\notin E_{s(n)}(S,\overline{S})}H_e$ on both sides of Eq.~(\ref{eqappbis}) to obtain that
\begin{equation}
\bigotimes_{e\in E_{s(n)}(S,\overline{S})}\rho_{e}(p,d)=q_S^{(n)}\eta_S^{(n)}+\sum_{M\neq S}q_M^{(n)}\omega_{M,S}^{(n)},
\end{equation}
where $\eta_S^{(n)}$ is separable in the bipartition $S|\overline{S}$. Then,
\begin{equation}
q_S^{(n)}\leq \mathrm{BSA}\left(\rho(p,d)^{\otimes |E_n(S,\overline{S})|}\right)
         \leq \mathrm{BSA}\left(\rho(p,d)^{\otimes n}\right)=:g(n)\, ,
\end{equation}
where we have now used that $|E_n(S,\overline{S})|\geq n$ for all $S\notin I^n$. Thus,
\begin{equation}
\sum_{M\notin I^n}q_M^{(n)}\leq\left(2^{f(n)-1}-1\right)g(n)=:h(n)\,.
\end{equation}
However, by Lemma~\ref{lem:tensor-bsa} we have that $\lim_{n\to\infty}g(n)=0$ for every choice of the visibility $p>1/(d+1)$. Therefore if $f$ has slow enough growth we can guarantee that $\lim_{n\to\infty}h(n)=0$ for any choice of $p\geq p_0$ where $p_0$ is given by the condition of Eq.~(\ref{argu1app}) above. Therefore, there exists a choice of $f$ such that, for $p$ sufficiently close to 1, the normalization condition
$\sum_Mq_M^{(n)}=1$ cannot hold for any $n$ that is large enough providing a contradiction.
\end{proof}

\section{Discussion and conclusions}\label{secconcl}

In this work we have drastically strengthened the observation of \cite{ASGME} that the AGME and ABS properties are possible in quantum network states by providing general sufficient conditions for both features based on well-established graph parameters. Although we have shown that, as intuition suggests, edge-connectivity plays an important role in this problem, it turns out that the degree growth provides stronger conditions.

Theorem~\ref{th:degABS} proves that all graph sequences with at most logarithmic degree growth must be ABS. This is in particular interesting because it shows that AGME requires that the degree in the network increases with its size and, therefore, it cannot be obtained with bounded local dimension. On the other hand, Theorem~\ref{th:degAGME} reveals that all graph sequences with linear degree growth are AGME. Thus, these two theorems provide an answer to our problem for very general classes of graphs. Nevertheless, this still leaves a considerable gap between logarithmic and linear growth. This issue is addressed by Theorems \ref{th:highdegABS} and \ref{th:lowdegAGME}, which show that the degree cannot characterize in general the AGME or the ABS properties. The first one proves that there always exist ABS graph sequences for any sublinear degree growth. Thus, the degree growth condition in Theorem \ref{th:degAGME} cannot be improved. On the other hand, Theorem~\ref{th:lowdegAGME} proves that AGME may exist at sublinear degree growth for any fractional power of choice. Additionally, this proof is constructive, so it provides instances of the cheapest networks that we know of in terms of the local dimension necessary to display the AGME property. Taking all this into account, the open question left in this respect is to characterize what the situation is between logarithmic and fractional power degree growth. Can we have AGME graph sequences $(G_n)_{n\in\N}$ for which $\delta_{\max}(G_n)$ grows polylogarithmically with $|G_n|$ and further reduce the local dimension overhead necessary to construct networks with this property or are all such graph sequences ABS?

Another possible future approach to this problem is to explore its relation to other graph parameters. In this sense, while most of the AGME graph sequences identified here have the property that $\textrm{diam}(G_n)\in O(1)$, we have showed in Theorem~\ref{teo:diam} that AGME graph sequences with unbounded diameter are also possible. However, this kind of questions can be explored in many other different directions. A particularly suggestive idea in principle is to consider the \textit{algebraic connectivity} or the so-called \textit{Cheeger constant}, which give rise to the notion of graph expander families (see e.g.\ \cite{chung}). These are graph sequences that can be regarded as highly connected and sparse. However, our results prove that these notions do not play a relevant role in the AGME/ABS question. Graph expanders can have bounded degree and, in fact, there exist
constructions of expander sequences of $k$-regular graphs with $k\geq 3$ (see \cite{marcus15} and references therein). Therefore, Theorem~\ref{th:degABS} entails that these objects do not necessarily lead to AGME. On the other hand, the graph sequences used in Proposition~\ref{th:lowlambda} have Cheeger constant going to zero as $n\to\infty$. Therefore, algebraic connectivity alone cannot give sufficient conditions for the ABS property either. However, other graph parameters might be suitable for capturing the behavior of multipartite entanglement addressed in this work.

Once the underlying dimension $d$ is fixed, if a graph sequence is AGME, this entails that there exists a value of the visibility $p<1$ such that the corresponding isotropic network state $\sigma_d(G_n,p)$ is GME for all $n$. However, the mere property of being AGME does not say anything about how far from 1 the corresponding visibility can be, which is a question of obvious practical relevance. Thus, it would be interesting to derive tools to address the problem of determining the threshold visibility of a given AGME graph sequence, i.e.\ the infimum over all $p_0$ for which the property that defines AGME as given in Definition~\ref{def:asymptotic} holds. Notice that in \cite{ASGME} we have argued that $\overline{p_2(K_n)}\lesssim0.865$ for $n$ sufficiently large, which is well within what is experimentally feasible. This indicates that the threshold visibility is likely to take accessible values as well for other less connected AGME graph sequences. Another feature that we have not touched upon here is the role of the dimension $d$ of the bipartite entanglement sources that give rise to the links in the network. The definitions of AGME and ABS depend on this choice; however, all the conditions that we obtained here are independent of this parameter. We leave for future study the question of whether there exists a graph sequence that is AGME or ABS (or none of the two) depending on the choice of $d$ or whether the dimension plays no role in this analysis.

In summary, we believe that the AGME/ABS problem in quantum networks unveils a mathematically rich scenario in the intersection of quantum information theory and graph theory. While we have shown here that these properties can be related to basic quantities in graph theory, no such notion seems to be able to completely characterize them. Therefore, we think that AGME leads to a novel graph parameter related to a different notion of connectivity. We hope that our work spurs further research in this direction.

\begin{acknowledgments}

We thank Gandalf Lechner and Zhen-Peng Xu for useful discussions. This work is partially supported by grant QUITEMAD-CM TEC-2024/COM-84 funded by Comunidad de Madrid, by grant CEX2023-001347-S funded by MCIN/AEI/10.13039/501100011033 and by grant PID2023-146758NB-I00 funded by the Spanish MICIU. F.L. thanks Wilhelm Winter for his kind invitation to the Excellence Cluster Mathematics M\"unster from March-July 2024, where part of this work was done.
Financial support was provided by the mobility grant PRX22/00472 by Ministerio de Universidades. He also acknowledges financial support from a grant 6G-INTEGRATION-3 (TSI-063000-2021-
127), funded by UNICO program (under the Next Generation EU umbrella funds), Ministerio de Asuntos Económicos y Transición Digital of Spain. F.L. and J.I. de V. are also supported by grant PID2024-160539NB-I00 funded by the Spanish MICIU. C.P. gratefully acknowledges financial support for this publication by the Fulbright Program, which is sponsored by the U.S. Department of State and the U.S.- Spain Fulbright Commission.

\end{acknowledgments}


\begin{thebibliography}{99}
\bibliographystyle{quantum}



\bibitem{percolation} A. Ac\'in, J. I. Cirac, and M. Lewenstein, Entanglement percolation in quantum networks, \href{https://doi.org/10.1038/nphys549}{Nature Physics \textbf{3}, 256 (2007)}.

\bibitem{BSAisotropic} S. J. Akhtarshenas and M. A. Jafarizadeh, Optimal Lewenstein–Sanpera decomposition for some bipartite
  systems, \href{https://doi.org/10.1088/0305-4470/37/8/008}{J. Phys. A: Math. Gen. \textbf{37}, 2965 (2004)}.
  
\bibitem{elkouss} K. Azuma, S. B{\"a}uml, T. Coopmans, D. Elkouss, and L. Boxi, Tools for quantum network design, \href{https://doi.org/10.1116/5.0024062}{AVS Quantum Sci. \textbf{3}, 014101 (2021)}.

\bibitem{beigi} S. Beigi and P. W. Shor, Approximating the set of separable states using the positive partial
  transpose test, \href{https://doi.org/10.1063/1.3364793}{J. Math. Phys. \textbf{51}, 042202 (2010)}.
  
\bibitem{bondy} J. A. Bondy and U. S. R. Murty, Graph Theory (Springer, New York, 2008).

\bibitem{pairable} S. Bravyi, Y. Sharma, M. Szegedy, and R. de Wolf, Generating k EPR-pairs from an n-party resource state, \href{https://doi.org/10.22331/q-2024-05-14-1348}{Quantum \textbf{8}, 1348 (2024)}.

\bibitem{distributedqc2} M. Caleffi, M. Amoretti, D. Ferrari, J. Illiano, A. Manzalini, and A. S. Cacciapuoti, Distributed quantum computing: A survey, \href{https://doi.org/10.1016/j.comnet.2024.110672}{Computer Networks \textbf{254}, 110672 (2024)}.
    
\bibitem{chartrand} G. Chartrand, A graph-theoretic approach to a communications problem, SIAM J. Appl. Math. \textbf{14}, 778 (1966).

\bibitem{locc} E. Chitambar, D. Leung, L. Mancinska, M. Ozols, and A. Winter, Everything you always wanted to know about LOCC (but were afraid to
  ask), \href{https://doi.org/10.1007/s00220-014-1953-9}{Commun. Math. Phys. \textbf{328}, 303 (2014)}.
  
\bibitem{chung} F. Chung, Spectral Graph Theory (American Mathematical Society, 1997).

\bibitem{distributedqc1} J. I. Cirac, A. Ekert, S. F. Huelga, and C. Macchiavello, Distributed Quantum Computation over Noisy Channels, \href{https://doi.org/10.1103/PhysRevA.59.4249}{Phys. Rev. A \textbf{59}, 4249 (1999)}.
    
\bibitem{networkGMNL} P. Contreras-Tejada, C. Palazuelos, and J. I. de Vicente, Genuine multipartite nonlocality is intrinsic to quantum networks, \href{https://doi.org/10.1103/PhysRevLett.126.040501}{Phys. Rev. Lett. \textbf{126}, 040501 (2021)}.
    
\bibitem{ASGME} P. Contreras-Tejada, C. Palazuelos, and J. I. de Vicente, Asymptotic survival of genuine multipartite entanglement in noisy
  quantum networks depends on the topology, \href{https://doi.org/10.1103/PhysRevLett.128.220501}{Phys. Rev. Lett. \textbf{128}, 220501 (2022)}.
  
\bibitem{gmekey} S. Das, S. B{\"a}uml, M. Winczewski, and K. Horodecki, Universal limitations on quantum key distribution over a network, \href{https://doi.org/10.1103/PhysRevX.11.041016}{Phys. Rev. X \textbf{11}, 041016 (2021)}.
    
\bibitem{devetakwinter} I. Devetak and A. Winter, Distillation of secret key and entanglement from quantum states, \href{https://doi.org/10.1098/rspa.2004.1372}{Proc. R. Soc. A \textbf{461}, 207 (2005)}.

\bibitem{diestel} R. Diestel, Graph Theory (Springer, New York, 2017).

\bibitem{erdos} P. Erd\H{o}s, J. Pach, R. Pollack, and Z. Tuza, Radius, diameter, and minimum degree, \href{https://doi.org/10.1016/0095-8956(89)90066-X}{J. Combin. Theory Ser. B \textbf{47}, 73 (1989)}.

\bibitem{hamada} M. Hamada, Exponential lower bound on the highest fidelity achievable by quantum error-correcting codes, \href{https://doi.org/10.1103/PhysRevA.65.052305}{Phys. Rev. A \textbf{65}, 052305 (2002)}.

\bibitem{isotropic} M. Horodecki and P. Horodecki, Reduction criterion of separability and limits for a class of
  distillation protocols, \href{https://doi.org/10.1103/PhysRevA.59.4206}{Phys. Rev. A \textbf{59}, 4206 (1999)}.
  
\bibitem{horodecki2009quantum} R. Horodecki, P. Horodecki, M. Horodecki, and K. Horodecki, Quantum entanglement, \href{https://doi.org/10.1103/RevModPhys.81.865}{Rev. Mod. Phys. \textbf{81}, 865 (2009)}.

\bibitem{gmesensing1} P. Hyllus, W. Laskowski, R. Krischek, C. Schwemmer, W. Wieczorek, H. Weinfurter, L. Pezz\'{e}, and A. Smerzi, Fisher information and multiparticle entanglement, \href{https://doi.org/10.1103/PhysRevA.85.022321}{Phys. Rev. A \textbf{85}, 022321 (2012)}.
    
\bibitem{BSA} M. Lewenstein and A. Sanpera, Separability and entanglement of composite quantum systems, \href{https://doi.org/10.1103/PhysRevLett.80.2261}{Phys. Rev. Lett. \textbf{80}, 2261 (1998)}.

\bibitem{marcus15} A. W. Marcus, D. A. Spielman, and N. Srivastava, Interlacing families I: Bipartite Ramanujan graphs of all degrees, \href{https://doi.org/10.4007/annals.2015.182.1.7}{Ann. Math. \textbf{182}, 307 (2015)}.
    
\bibitem{Nielsen} M. A. Nielsen and I. L. Chuang, Quantum Computation and Quantum Information (Cambridge University Press, Cambridge, 2010).

\bibitem{percolationreview} S. Perseguers, G. J. Lapeyre Jr, D. Cavalcanti, M. Lewenstein, and A. Ac\'{i}n, Distribution of entanglement in large-scale quantum networks, \href{https://doi.org/10.1088/0034-4885/76/9/096001}{Rep. Prog. Phys. \textbf{76}, 096001 (2013)}.
    
\bibitem{pleniovirmani} M. B. Plenio and S. Virmani, An introduction to entanglement measures, \href{https://doi.org/10.26421/QIC7.1-2-1}{Quantum Inf. Comput. \textbf{7}, 1 (2007)}.

\bibitem{kraftspee} C. Spee and T. Kraft, Transformations in quantum networks via local operations assisted by finitely many rounds of classical communication, \href{https://doi.org/10.22331/q-2024-03-14-1286}{Quantum \textbf{8}, 1286 (2024)}.
    
\bibitem{gmesensing2} G. T\'{o}th, Multipartite entanglement and high-precision metrology, \href{https://doi.org/10.1103/PhysRevA.85.022322}{Phys. Rev. A \textbf{85}, 022322 (2012)}.

\bibitem{Watrous} J. Watrous, The Theory of Quantum Information (Cambridge University Press, Cambridge, 2018).

\bibitem{wehner2018quantum} S. Wehner, D. Elkouss, and R. Hanson, Quantum internet: A vision for the road ahead, \href{https://doi.org/10.1126/science.aam9288}{Science \textbf{362}, 9288 (2018)}.

\bibitem{yamasaki} H. Yamasaki, A. Soeda, and M. Murao, Graph-associated entanglement cost of a multipartite state in exact and finite-block-length approximate constructions, \href{https://doi.org/10.1103/PhysRevA.96.032330}{Phys. Rev. A \textbf{96}, 032330 (2017)}.

\end{thebibliography}
\end{document}